\documentclass[review]{elsarticle}
 
\usepackage{lineno,hyperref}
\modulolinenumbers[5]
\usepackage[utf8]{inputenc}
\usepackage[T1]{fontenc}
\usepackage{color}
\definecolor{note_fontcolor}{rgb}{0.800781, 0.800781, 0.800781}
\usepackage{verbatim}
\usepackage{float}
\usepackage{amsmath}
\usepackage{amsthm}
\usepackage{amssymb}
\usepackage{graphicx}
\usepackage{subfig}

\usepackage{color}
\usepackage{comment}
\usepackage{ifthen}
\usepackage{mathrsfs}
\usepackage{wrapfig}

\newtheorem{property}{Property}
\newtheorem{lem}{Lemma}
\newtheorem{exmple}{Example}

\journal{Computer Networks}

\begin{document}
\begin{frontmatter}
\title{MICN: a network coding protocol for ICN\\
with multiple distinct interests per generation}

\author[1,3]{H. Malik\corref{cor1}}
\ead{hirah.malik@inria.fr}
\author[1]{C. Adjih}
\ead{cedric.adjih@inria.fr}
\author[2]{C. Weidmann}
\ead{claudio.weidmann@ensea.fr}
\author[3]{ M. Kieffer}
\ead{michel.kieffer@l2s.centralesupelec.fr}

\cortext[cor1]{Corresponding author}
\address[1]{Inria Saclay, 1 Rue Honoré d'Estienne d'Orves, 91120 Palaiseau, France}
\address[2]{ETIS UMR8051, CY University, ENSEA, CNRS, 95000 Cergy, France}
\address[3]{Université Paris-Saclay - CNRS - CentraleSup\'elec, Laboratoire des Signaux et Syst\`emes,\\ 3 rue Joliot-Curie, 91192 Gif-sur-Yvette, France}

\begin{abstract}
In Information-Centric Networking (ICN), consumers send interest packets
to the network and receive data packets as a response to their request
without taking care of the producers, which have provided the content,
contrary to conventional IP networks. ICN supports the use of multiple
paths; however, with multiple consumers and producers, coordination
among the nodes is required to efficiently use the network resources.
Network coding (NC) is a promising tool to address this issue. The
challenge in the case of NC is to be able to get independent coded
content in response to multiple parallel interests by one or several
consumers. In this work, we propose a novel construction called MILIC
(Multiple Interests for Linearly Independent Contents) that impose
constraints on how the replies to interests are coded, intending to
get linearly independent contents in response to multiple interests.
Several protocol variants, called MICN (MILIC-ICN), built on top of
NDN (Named Data Networking), are proposed to integrate these interest
constraints and NC of data packets. Numerical analysis and simulations
illustrate that the MILIC construction performs well and that the
MICN protocols are close to optimal throughput on some scenarios.
MICN protocols compare favorably to existing protocols, and show
significant benefits when considering the total number of transmitted
packets in the network, and in the case of high link loss rate.
\end{abstract}

\begin{keyword}
Information centric networking \sep ICN \sep Network coding \sep Named Data Networking
\end{keyword}
\end{frontmatter}

\section{Introduction}

Content distribution has become the primary task for today's Internet.
According to CISCO's forecast, video traffic will be accounting for
79 percent of total mobile data traffic by 2022 \cite{Forecast2019}.
The communication network's traditional paradigm has some drawbacks,
especially when dealing with large-scale content distribution because
of the point-to-point nature of communication and location dependence.
The consumers, however, care about the content itself and not about
its origin.

Information-Centric Networking (ICN) has recently been proposed as
an alternative to the traditional point-to-point communication to
make content the center of the communication network \cite{Jacobson2009}.
The ICN principle is based on receiving data through names by performing
a named-based routing. It removes the need to establish a connection
between endpoints and allows caching throughout the network. Named
data networking (NDN) \cite{Afanasyev2018,Zhang2010} is one of the
ICN architectures.

The basic NDN framework is a pull-based mechanism where the clients
send interest packets that contain the name of the requested content.
These interest packets are routed based on their names. A node holding
a copy of the requested content replies to the interest with the content
in a data packet.

Devices nowadays come with multiple network interfaces that can be
used to retrieve content, e.g., WiFi, 3G/LTE. Traditional networking
requires to establish a session among endpoints and hence does not
allow simultaneous use of all available interfaces. NDN, however,
enables the use of multiple interfaces. Nevertheless, in a dynamic
network with multiple clients, some coordination is required to take
advantage of multiple paths. Montpetit \emph{et al}. introduced an
alternative to the coordination approach by utilizing network coding
(NC) over NDN \cite{Montpetit2012}.

NC is a communication paradigm which, unlike traditional networking,
allows the nodes to perform operations on the packets (computing algebraic
combination of packets) \cite{Ahlswede2000,Koetter2003}. Decoding is performed
by solving a linear system of equations once enough linearly independent
combinations/packets are received. NC helps to exploit the network's
capacity, minimize delays and may help recover from link failures.
Traditional routers that could only forward or replicate the packets
are replaced by coding routers that can mix packets of the same content.

In this work, we aim to integrate more efficiently NC within NDN.
For that purpose, special construction of the interests called MILIC
(Multiple Interests for Linearly Independent Contents) is introduced
that imposes some constraints on the content these interests bring.
Several protocol variants, called MICN (MILIC-ICN), are then built
on top of MILIC and NDN. MICN protocols allow parallel processing
of multiple interests send by nodes and ensure linearly independent
content with each of the multiple interests. Numerical analysis and
simulations illustrate that the MILIC construction performs well and
that the MICN protocols manage to get close to the optimal throughput
on the considered scenarios. The performances obtained with MICN compare
favorably to existing protocols and show significant benefits when
considering the total number of transmitted packets in the network,
and in the case of high link loss rate.

Section~\ref{Sec:Related-Work} summarizes some related work. Section~\ref{sec:Our-approach}
details the special construction of the interests MILIC. In Section~\ref{sec:Protocol},
we detail the MICN protocol that uses MILIC construction to integrate
NC over NDN. Section~\ref{sec:Optimizations} introduces some optimizations
to improve the performance of the protocol further. Section~\ref{sec:Evaluation}
presents the simulation setup and results. Section~\ref{sec:Conclusion}
concludes this paper and introduces future work.

\section{Related Work}

\label{Sec:Related-Work}

\subsection{NDN}

In this section, we briefly explain the basic concepts of the NDN
architecture \cite{Afanasyev2018,Zhang2010}. Communication in NDN
is \emph{consumer-driven,} with two basic types of communication packets:
\emph{interest} and \emph{data }packets. Consider an NDN network consisting
of a set $\mathcal{N}$ of nodes. Nodes can be \emph{sources} that
generate content, intermediate nodes or \emph{caching routers}, or
\emph{client}s that request content. A node can have any of these
roles at a given time. Each node $r\in\mathcal{N}$ is connected in
the network through a set of faces $\mathcal{F}_{r}$. The term \emph{face}
is a generalization of the interface that corresponds to various communication
links.

The clients request the network to find the content by sending an
NDN interest packet that carries the name of the requested content.
The interest is forwarded in the network until it reaches a node holding
a copy of the content with the requested name. The content is then
sent back in a data packet. Both interest and data packets carry the
name of the content but there is no information regarding the client
or source.

Each NDN node has a \emph{Pending Interest Table} (PIT), a \emph{Forwarding
Information Base} (FIB), and a \emph{Content Store} (CS) for the transport
of the named content in the network \cite{Jacobson2009}. The PIT
keeps a record of pending interests forwarded by the node that are
not satisfied yet. Along with interest, the PIT stores the face where
each interest arrives (in-faces) and the faces to which it was forwarded
(out-faces) to record the reverse link for the data packet. The FIB
stores routing information used to forward interest packets toward
potential sources of matching content ; it can be populated by self-learning
or by a routing protocol. Routing and forwarding strategies
to efficiently perform the named based routing are presented in \cite{Ahmed2016} and \cite{Posch2017}. Finally, the CS is
a cache memory. An intermediate node can decide to cache the content
that it forwards downstream for replying to future interests. Consequently,
content is stored in source nodes and in caching routers \cite{Afanasyev2018}.

Each interest packet brings back one data packet if a copy of the
requested content is found. If the content object is large, it may
be partitioned into smaller segments to fit into data packets. In
a classical NDN, the client requests a content segment by sending
the name prefix with the segment identifier \cite{Gusev2015}. For
example, the interest \texttt{\small{}<content-name>/<}{\small{}$i$}\texttt{\small{}>
}is requesting the $i^{th}$ segment of a content. Each interest also
carries a random identifier, named \emph{nonce},\emph{ }which helps
to prevent interest forwarding loops \cite{Afanasyev2018}. Interests
are forwarded using the information in the FIB.

A node that receives interest for a segment verifies that there is
no similar interest pending in its PIT. Having a pending interest
means that the requested content is not in the CS, so the node updates
the pending entry in the PIT by adding the receiving face of the new
interest. If there is no pending entry with the same name, the node
then checks its CS for a copy of the requested content. If there is
a \emph{cache hit,} \emph{i.e.}, the requested content is available
in the cache, the node replies to the interest with a data packet.
In the case of the unavailability of requested content or \emph{cache
miss}, a new PIT entry is created, and the interest is forwarded to
the available faces in the FIB. Once the content is received from
upstream, it is routed back to the requesting node using the information
in the PIT. The node also decides whether the content should be cached
locally \cite{Afanasyev2018}.

\subsection{Network Coding and NDN}

NC and NDN both inherently tend to address the content delivery and
focus on the improvement of content distribution over the network.
NC and NDN can work jointly to exploit network capacity better (by
exploiting caching, multi-path delivery, \emph{etc}.). The idea of
applying NC over ICN/NDN was first introduced in NC3N by Montpetit
\emph{et al.} to take advantage of NDN and NC's inherent features
to improve content delivery \cite{Montpetit2012}. Currently, there is ongoing
work in standardization on the precise topic of mixing NC and ICN \cite{irtf-nwcrg-nwc-ccn-reqs-03}.

In an NC scenario, the original content is partitioned into smaller
groups of segments, called \emph{generations. }NC is only allowed
among the segments of a generation to reduce the decoding complexity.
In a given generation, segments may be linearly combined within a
source or at any intermediate node in the network. The linear combinations
are performed in some Galois field $\mathbb{F}_{q}$ to get \emph{coded}
segments. The coefficients in $\mathbb{F}_{q}$ of the linear combination
form the \emph{encoding vector }\cite{Chou2003} of each coded segment.
In the NDN context, the source nodes and caching routers may store
original and coded segments. 

The client nodes send interests requesting coded segments instead
of a specific segment. The name carried by interest packets for coded
content is adapted to indicate that a coded content is expected (\emph{e.g.},
by setting a flag that indicates the retrieval of coded segments \cite{Montpetit2012}).
Requesting content like this allows the intermediate or source nodes
to send different coded segments generated by combining the original
segments of the content in their cache, instead of one particular
segment.

The pull-based request and response mechanism of NDN allows one interest
to bring back one content segment. A client node sends at least $n$
interests to be able to decode a generation of $n$ segments. Based
on the interest processing in NDN, it is challenging to ensure retrieval
of \emph{innovative} (linearly independent) content with each interest
as is required to keep minimal decoding delay and network load.

In the coded NDN schemes proposed in \cite{Lei2017b,Wu2013,Wu2016a},
encoding vectors of all the received coded segments are sent in the
interests. The encoding vectors help the nodes to either generate
coded content that will be innovative for the requesting node or forward
the interest to their next-hop neighbors. However, this approach introduces
an overhead in the interest packets that increases with each coded
segment the requesting node receives. The size of the overhead is
limited by keeping the generation size small. This approach also introduces
a delay as the client node waits to receive the replies for previously sent
interests to arrive before it can issue interests for more coded segments
to ensure retrieval of linearly independent content with each interest.

Zhang \emph{et al}.compare the approach of sending all received coefficients
(precise matching) to rank-based matching, \emph{i.e.}, sending only
the client node's rank. They observe that rank-based matching achieves
slightly lower performance but has much lower computation and communication
overhead \cite{Zhang2016}.

Liu \emph{et al}. \cite{Liu2016} introduced an interest coding and
forwarding strategy that allow splitting and joining of interests
for the same content and generation at intermediate nodes. The interests
request a subset of segments by indicating the number of required
coded segments to get a decodable generation. This scheme implements
a one interest-multiple replies strategy, which is contrary to the
NDN principle of one interest-one reply.

NetCodCCN \cite{Saltarin2016}, tries to address the shortcomings
of previous approaches by sending undifferentiated interests for coded
segments of a generation. The client node implicitly states that it
requires another coded combination by sending additional interests
for coded segments. The intermediate routers that have previously
sent coded segments keep track of the number of coded segments forwarded
on each face and the rank of the set of linear combinations in their
CS. The node only replies to interest if the rank of its CS is bigger
than the number of coded segments it has sent on a particular face.
NetCodCCN also supports the transmission and handling of multiple
interests at one time (pipelining), to allow nodes to request content
more efficiently. With pipelining, a burst of interests is sent first
by a client. Each time content is received, a new interest is sent.

This approach increases the amount of information stored per router.
The other disadvantage of this approach is unnecessary data traffic
in the network that flows in the network after the clients have received
a decodable generation.

Liu \emph{et al}. \cite{Liu2017a} add a parameter\emph{ $r$} indicating
the number of desired coded segments in the interest requesting more
than one coded segment, which is again contrary to the NDN principle
of one interest-one reply. Matsuzono \emph{et al}. \cite{Matsuzono2017a}
proposed L4C2 a low loss, low latency, network coding enabled video
streaming over CCN. In L4C2 nodes, request network coded packets in
case of data packet losses only. Bilal \emph{et al}. \cite{Bilal2018a}
proposed an algebraic framework of Network coding over NDN. 

\section{Our approach for NC-NDN}

\label{sec:Our-approach}

In the classical NDN framework, there is a one-to-one mapping between
the requested content and interests. For a content split in $n$ segments,
the client sends $n$ interests, each requesting a specific segment
by stating its unique name and segment number. The replies to these
interests are the requested segments.

With NC, coded segments are stored at various places of the network.
As observed, \emph{e.g.}, in \cite{Montpetit2012}, many different
coded segments could serve as a reply to a given interest. The main
challenge when sending interests for coded segments is to ensure that
linearly independent segments are sent as replies. This problem is
challenging since any intermediate node that has cached coded segments
from a generation can generate one or more coded segments and hence
can reply to several interests from the same client, but linear independence
among the replies is not ensured.

As mentioned, in some prior work \cite{Lei2017b,Wu2013,Wu2016a},
only one outstanding interest is allowed at the expense of delay.
Alternately, \cite{Liu2017a} reduces the delay and overhead by requesting
multiple coded segments in one interest. Saltarin \emph{et al. }in
\cite{Saltarin2016,Saltarin2017a} overcome the delay problem by pipelining
multiple identical interests, at the cost of additional information
that needs to be stored at intermediate nodes in the network for proper
interest processing. 

In this work, we propose a NC-based NDN protocol with minimum overhead.
The client nodes indicate to the network what they require to decode
a generation by pipelining multiple distinct interests. The distinct
interests allow parallel processing of the multiple interests and
ensure that replies to each of the pipelined interests are not redundant.

\subsection{MILIC\label{subsec:MILIC}}

The main idea of MILIC starts from this previously described pipelining
idea, where a subset of interests for a content generation $g$ are
sent in a burst, with the goal that each of them brings innovative
coded content.

For a generation of size $n$, we propose to use $n$ distinct interests.
Interest $i\in\left\{ 1,\ldots,n\right\} $ can be satisfied by any
coded segment whose encoding vector belongs to a predefined subset
$\mathcal{A}_{i}$ of the set of all possible encoding vectors. In
the following, these constraints will be such that the set of all
non-zero encoding vectors will be partitioned into $n$ non-overlapping
subsets $\mathcal{A}_{1},\dots,\mathcal{A}_{n}$ satisfying some additional
constraints.

First, to ensure that the answer to any of the $n$ interests is linearly
independent of the other answers, the subsets must satisfy the following
property.
\begin{property}
\label{eq:LinIndep}For any $a_{1}\in\mathcal{A}_{1},\dots,a_{n}\in\mathcal{A}_{n}$,
the vectors $a_{1},\dots,a_{n}$ should be linearly independent, \emph{i.e.},
$\sum_{i=1}^{n}\alpha_{i}a_{i}=0\text{ iff }\alpha_{1}=\dots=\alpha_{n}=0.$
\end{property}
An additional condition can be imposed on subsets $\mathcal{A}_{1},\dots,\mathcal{A}_{n}$
to benefit from the observation that when a node sends the same interests
over $\ell$ faces, $\ell$ answers to each of these interests will
likely be received. Ideally, these replies should be linearly independent.
This leads to a property of subsets that is not mandatory but desirable
to improve the efficiency of the proposed solution.
\begin{property}
\label{eq:Rank}Consider $k$ distinct subsets $\mathcal{A}_{\pi\left(1\right)},\dots,\mathcal{A}_{\pi\left(k\right)}$.
Consider $\ell\geqslant1$ vectors $a_{\kappa}^{1},\dots,a_{\kappa}^{\ell}$
chosen uniformly at random from each subset $\mathcal{A}_{\pi\left(\kappa\right)}$,
$\kappa=1,\dots,k$ such that $\ell k\leqslant n$, then
$\text{rank}\left(a_{1}^{1},\dots,a_{k}^{\ell}\right)=\ell k $ with high probability.
\end{property}
Finally, one may try to exploit the fact that segments are coded with
possible re-encoding at intermediate nodes. Intermediate nodes may
have received several segments belonging to the same subset. It may
be of interest to combine these to generate a coded segment belonging
to another subset to satisfy interest for that subset. This translates
into the following additional desirable property for the subsets.
\begin{property}
\label{eq:LinComb}Consider the subset $\mathcal{A}_{i}$, $i\in\left\{ 1,\dots,n-1\right\} $.
For any pair $\left(a_{i}^{1},a_{i}^{2}\right)$ of linearly independent
vectors belonging to $\mathcal{A}_{i}$, then with high probability,
there exist $\alpha_{1}\in\mathbb{F}_{q}^{*}$ and $\alpha_{2}\in\mathbb{F}_{q}^{*}$
such that $ \alpha_{1}a_{i}^{1}+\alpha_{2}a_{i}^{2}\in\mathcal{A}_{k} $
with $k\neq i$.
\end{property}

\subsection{Proposed construction}

Here we propose a construction of the sets $\mathcal{A}_{1},\dots,\mathcal{A}_{n}$,
called MILIC, that partly satisfies the above properties. Consider
\begin{equation}
\mathcal{A}_{i}=\left\{ \left(v_{1},\ldots,v_{n}\right)\in\mathbb{F}_{q}^{n}\,|\,v_{i}\neq0\text{{\:and\:}}\forall j<i,v_{j}=0\right\} ,\label{eq:MILICSets}
\end{equation}
with $i=1,\dots,n$. With this construction the sets $\mathcal{A}_{1},\dots,\mathcal{A}_{n}$
form a partition of $\mathbb{\ensuremath{F}}_{q}^{n}\setminus\left\{ \left(0,\dots,0\right)\right\} $.
The construction (\ref{eq:MILICSets}) implicitly imposes an ordering
among sets, when considering their cardinal number. 

We will now prove that the proposed MILIC construction satisfies Property~\ref{eq:LinIndep},
Property~\ref{eq:Rank} for the $k$ first subsets $\mathcal{A}_{1},\dots,\mathcal{A}_{k}$,
and Property~\ref{eq:LinComb} for all $k>i$. 

Property~\ref{eq:LinIndep} is satisfied by construction: consider
any $a_{1}\in\mathcal{A}_{1},\dots,a_{n}\in\mathcal{A}_{n}$. The
matrix whose rows are $a_{1},\dots,a_{n}$ is in row echelon form,
and thus of full rank. The vectors $a_{1},\dots,a_{n}$ are thus linearly
independent.
\newtheorem{thm}{Theorem}
Then, we have the following property of the size of each subset $A_{k}$.

\newtheorem{prop}[thm]{Proposition}
\begin{prop}
The cardinal number of $\mathcal{A}{}_{k}$ verifies $\vert\mathcal{A}{}_{k}\vert=\left(q-1\right)q^{n-k}$.
\end{prop}

\begin{proof}
Consider first the case of $\mathcal{A}_{1}$: $\forall a_{i}\in\mathcal{A}_{1}$,
one has $a_{i,1}\neq0.$ There are $q^{n-1}$ vectors with leading
zeros in $\mathbb{F}_{q}^{n}$ hence $\vert\mathcal{A}_{1}\vert=\left(q-1\right)q^{n-1}.$
Then consider the more general case of $\mathcal{A}{}_{k}$ for $k>1$:
$\forall a_{i}\in\mathcal{A}_{k}$, one has $a_{i,j}=0$ for $\;j=1,\ldots,k-1$
and $a_{i,k}\neq0$. For $a_{i,k}$, we have $q-1$ possible choices.
Then each $a_{i,j}$, $\;j=k+1,\ldots,n$ may take $q$ possible values.
Consequently $\left(a_{i,k+1},\dots,a_{i,n}\right)$ may take $q^{n-k}$
possible values and $\vert\mathcal{A}{}_{k}\vert=\left(q-1\right)q^{n-k}$.
\end{proof}
We start proving Property~\ref{eq:Rank} for a single subset $\mathcal{A}_{k}$
provided that $\ell\leqslant n-k+1$, evaluating the probability of
having $\text{rank}\left(a_{k}^{1},\dots,a_{k}^{\ell}\right)=\ell$.

\begin{lem}
Consider $\ell$ vectors $a_{k}^{1},\dots,a_{k}^{\ell}$ chosen uniformly
at random from the set \textup{$\mathcal{A}_{k}$}, $k=1,\ldots,n$,
and with $1\leq\ell\leq n$. The probability that $a_{k}^{1},\dots,a_{k}^{\ell}$
are linearly independent is \textup{
\[
\Pr\left(\text{rank}\left(a_{k}^{1},\dots,a_{k}^{\ell}\right)=\ell\right)=\prod_{\ell=1}^{\ell}\left(1-\frac{q^{\ell-1}-1}{\left(q-1\right)q^{n-k}}\right).
\]
}
\end{lem}

\begin{proof}
Consider first $\ell=2$, and $a_{k}^{1}\in\mathcal{A}_{k}$. The
set of non-zero vectors collinear to $a_{k}^{1}$ and included in
$\mathcal{A}_{k}$ is $\text{span}\left(a_{k}^{1}\right)\cap\mathcal{A}_{k}=\text{span}\left(a_{k}^{1}\right)\setminus\{\left(0,\dots,0\right)\}$,
whose size is $q-1$. When choosing a second vector $a_{k}^{2}\in\mathcal{A}_{k}$
uniformly at random, the probability that $a_{k}^{1}$ and $a_{k}^{2}$
are linearly dependent is equal to the probability that $a_{k}^{2}\in\text{span}\left(a_{k}^{1}\right)\setminus\{\left(0,\dots,0\right)\}$.
Consequently, the probability that $a_{k}^{1}$ and $a_{k}^{2}$ are
linearly independent is
\begin{align*}
\Pr\left(\text{rank}\left(a_{k}^{1},a_{k}^{2}\right)=2\right) & =1-\frac{\vert\text{span}\left(a_{k}^{1}\right)\setminus\{\left(0,\dots,0\right)\}\vert}{\vert\mathcal{A}{}_{k}\vert}\\
 & =1-\frac{1}{q^{n-k}}.
\end{align*}
 Assume now that the $j-1$ first vectors $a_{k}^{1}\in\mathcal{A}_{k},\ldots,a_{k}^{j-1}\in\mathcal{A}_{k}$
are linearly independent. The set of vectors that are linearly dependent
with $a_{k}^{1},\ldots,a_{k}^{j-1}$ and included in $\mathcal{A}_{k}$
is $\text{Span}\left(a_{k}^{1},\ldots,a_{k}^{j-1}\right)\cap\mathcal{A}_{k}=\text{span}\left(a_{k}^{1},\ldots,a_{k}^{j-1}\right)\setminus\{\left(0,\dots,0\right)\}$.
Its size is $q^{j-1}-1$. Then, when choosing $a_{k}^{j}\in\mathcal{A}_{k}$
uniformly at random, the probability that $a_{k}^{1},\ldots,a_{k}^{j}$
are linearly dependent is equal to the probability that $a_{k}^{j}\in\text{span}\left(a_{k}^{1},\ldots,a_{k}^{j-1}\right)\setminus\{\left(0,\dots,0\right)\}$.
Consequently
\begin{align}
\Pr\left(\text{rank}\left(a_{k}^{1},\ldots,a_{k}^{j}\right)=j\,|\,\text{rank}\left(a_{k}^{1},\ldots,a_{k}^{j-1}\right)=j-1\right) & =1-\frac{q^{j-1}-1}{\left(q-1\right)q^{n-k}}.\label{eq:PrRank}
\end{align}
Then one has
\begin{align}
\Pr\left(\text{rank}\left(a_{k}^{1},\dots,a_{k}^{\ell}\right)=\ell\right) & \nonumber \\ &\hspace{-1cm}=\Pr\left(\text{rank}\left(a_{k}^{1},\dots,a_{k}^{\ell}\right)=\ell,\text{rank}\left(a_{k}^{1},\dots,a_{k}^{\ell-1}\right)=\ell-1\right)\nonumber \\
 &\hspace{-1cm}=\Pr\left(\text{rank}\left(a_{k}^{1},\dots,a_{k}^{\ell}\right)=\ell\,|\,\text{rank}\left(a_{k}^{1},\dots,a_{k}^{\ell-1}\right)=\ell-1\right)\nonumber \\
 & \Pr\left(\text{rank}\left(a_{k}^{1},\dots,a_{k}^{\ell-1}\right)=\ell-1\right).\label{eq:chainRule}
\end{align}
Applying this recursively and using (\ref{eq:PrRank}), one gets
\begin{align*}
\Pr\left(\text{rank}\left(a_{k}^{1},\dots,a_{k}^{\ell}\right)=\ell\right) & \\ &\hspace{-2cm}=\prod_{j=2}^{\ell}\Pr\left(\text{rank}\left(a_{k}^{1},\dots,a_{k}^{j}\right)=j\,|\,\text{rank}\left(a_{k}^{1},\dots,a_{k}^{j-1}\right)=j-1\right)\\
 & \Pr\left(\text{rank}\left(a_{k}^{1}\right)=1\right)\\
 & \hspace{-2cm}=\prod_{j=1}^{\ell}\left(1-\dfrac{q^{j-1}-1}{\left(q-1\right)q^{n-k}}\right).
\end{align*}
\end{proof}

\begin{exmple}
Table~\ref{tab:error Probability} provides $P_{\text{F}}\left(\ell,1\right)=1-\Pr\left(\text{rank}\left(a_{k}^{1},\dots,a_{k}^{\ell}\right)=\ell\right)$
for vectors of $n=10$ elements in $\mathbb{F}_{256}$ for different
subsets $\mathcal{A}_{k}$ and different values of $\ell$. One observes
that choosing $5$ vectors at random from any of the subsets $\mathcal{A}_{k}$,
$k=1,\dots,5$, results in a very high probability of getting linearly
independent vectors. Consequently, if a client sends $5$ interest
packets for elements in $\mathcal{A}_{k}$ over different faces, it
is likely, provided that these interests follow different paths in
the network, to get $5$ linearly independent data packets.
\end{exmple}

\begin{table}[H]
\centering{}%
\begin{tabular}{|c|c|c|c|c|c|}
\hline 
 & $\ell=1$ & $\ell=2$ & $\ell=3$ & $\ell=4$ & $\ell=5$\tabularnewline
\hline 
\hline 
$\mathcal{A}_{1}$ & $0$ & $2.11\times10^{-22}$ & $5.46\times10^{-20}$ & $1.39\times10^{-17}$ & $3.58\times10^{-15}$\tabularnewline
\hline 
$\mathcal{A}_{2}$ & $0$ & $5.46\times10^{-20}$ & $1.39\times10^{-17}$ & $3.58\times10^{-15}$ & $9.16\times10^{-13}$\tabularnewline
\hline 
$\mathcal{A}_{3}$ & $0$ & $1.39\times10^{-17}$ & $3.58\times10^{-15}$ & $9.16\times10^{-13}$ & $2.34\times10^{-10}$\tabularnewline
\hline 
$\mathcal{A}_{4}$ & $0$ & $3.58\times10^{-15}$ & $9.16\times10^{-13}$ & $2.34\times10^{-10}$ & $6.01\times10^{-8}$\tabularnewline
\hline 
$\mathcal{A}_{5}$ & $0$ & $9.16\times10^{-13}$ & $2.34\times10^{-10}$ & $6.01\times10^{-8}$ & $1.53\times10^{-5}$\tabularnewline
\hline 
\end{tabular}\caption{\label{tab:error Probability}Probability of getting linearly dependent
coded vectors chosen at random from $\mathcal{A}_{k}\subset\mathbb{F}_{256}^{10}$}
\end{table}

We now prove Property~\ref{eq:Rank} for the $k$ first subsets $\mathcal{A}_{1},\dots,\mathcal{A}_{k}$.
\begin{lem}
Consider $\ell\geqslant1$ vectors $a_{\kappa}^{1},\dots,a_{\kappa}^{\ell}$
chosen uniformly at random from each subset $\mathcal{A_{\kappa}}$,
$\kappa=1,\dots,k$ such that \textup{$\ell k\leqslant n$}. The probability
that $a_{1}^{1},\dots,a_{k}^{\ell}$ are linearly independent is 
\[
\Pr\left(\text{rank}\left(a_{1}^{1},\dots,a_{k}^{\ell}\right)=\ell k\right)=\prod_{j=1}^{\left(\ell-1\right)k}\left(1-\frac{q^{j-1}}{q^{n-k}}\right).
\]
\end{lem}

\begin{proof}
According to Property~1, the vectors $a_{1}^{1},\dots,a_{k}^{1}$
are linearly independent. Consider the matrix $A$, whose first $k$
rows are the vectors $a_{1}^{1},\dots,a_{k}^{1}$ and the $\left(\ell-1\right)k$
remaining rows are $a_{\kappa}^{2},\dots,a_{\kappa}^{\ell}$, $\kappa=1,\dots,k$.
The first $k$ rows are used to perform Gaussian elimination on the
$\left(\ell-1\right)k$ remaining rows to get a matrix $A_{1}$ of
the form

\[
A_{1}=\left[\begin{array}{ccccccc}
1 & * & \cdots &  &  &  & *\\
0 & 1 & \ddots &  &  &  & \vdots\\
\vdots & \ddots & \ddots\\
 &  &  & 1 & * & \cdots & *\\
 &  &  & 0\\
\vdots &  &  & \vdots &  & B\\
0 & \cdots &  & 0
\end{array}\right].
\]
In $A_{1}$, $B$ is a matrix of $\left(\ell-1\right)k$ rows and
$n-k$ columns. Since all vectors chosen in the subset $\mathcal{A_{\kappa}}$,
$\kappa=1,\dots,k$, have been selected uniformly at random, the $n-k$
last entries of each vector are independently and uniformly distributed.
The $i$-th row of $B$ results in a linear combination of $a_{1}^{1},\dots,a_{k}^{1}$
with one of the remaining vectors $a_{\kappa}^{2},\dots,a_{\kappa}^{\ell}$,
$\kappa=1,\dots,k$. Consequently, the $n-k$ components of the $i$-th
row of\textbf{ $B$ }are still independently and uniformly distributed.
Since all $n-k$ last components of $a_{\kappa}^{2},\dots,a_{\kappa}^{\ell}$,
$\kappa=1,\dots,k$ are independently and uniformly distributed; all
components of the matrix $B$ are independently and uniformly distributed.

The matrix $A$ is of full row rank $\ell k$ iff the matrix $B$
is full row rank $\left(\ell-1\right)k$. The first row $b_{1}\in B$\textbf{
}is non-zero with probability $1-\frac{1}{q^{n-k}}$. The second row
$b_{2}\in B$ has components that are uniformly and independently
distributed from the other entries of $B$ and thus of $b_{1}$. The
vectors $\left(b_{1},b_{2}\right)$ are linearly independent if $b_{2}$
does not belong to the space spanned by $b_{1}$. Since $\text{span}\left(b_{1}\right)$
is of size $q$, one has 
\begin{align*}
\Pr\left(\text{rank}\left(b_{1},b_{2}\right)=2\right) & =1-\frac{q}{q^{n-k}}.
\end{align*}

Assume now that the $j-1$ first row vectors $b_{1},\ldots,b_{j-1}$
of $B$ are linearly independent. Under this assumption, the probability
that $b_{j}$ is such that the $j$ first row vectors $b_{1},\ldots,b_{j}$
of $B$ are linearly independent is equal to the probability that
$b_{j}$ does not belong to the subspace of dimension $q^{j-1}$ spanned
by $b_{1},\ldots,b_{j-1}$. Consequently, 
\begin{align*}
\Pr\left(\text{rank}\left(b_{1},\ldots,b_{j}\right)=j\,|\,\text{rank}\left(b_{1},\ldots,b_{j-1}\right)=j-1\right) & =1-\frac{q^{j-1}}{q^{n-k}}.
\end{align*}
Then similarly as \ref{eq:chainRule}, the probability that $B$ is
of full rank is given by
\begin{align*}
\Pr\left(\text{rank}\left(B\right)=\left(\ell-1\right)k\right) & =\prod_{j=1}^{\left(\ell-1\right)k}\left(1-\frac{q^{j-1}}{q^{n-k}}\right).\\
&= \prod_{j=1}^{\left(\ell-1\right)k}\left(1-\frac{1}{q^{n-k-j+1}}\right)\\
&\approx 1-\frac{{1}}{q^{n-lk+1}}\text{{when} \ensuremath{q} large}
\end{align*}
\end{proof}
\begin{exmple}
Table~\ref{tab:probability-k-sets} provides $P_{\text{F}}\left(\ell,k\right)=1-\Pr\left(\text{rank}\left(a_{1}^{1},\dots,a_{k}^{\ell}\right)=\ell k\right)$
for vectors of $n$ elements in $\mathbb{F}_{q}$ when choosing at
random $\ell$ vectors from each subset $\mathcal{A_{\kappa}}$, $\kappa=1,\dots,k$.
One observes that when a node receives $2$ random packets from each
of the $\mathcal{A_{\kappa}}$, $\kappa=1,\dots,50$ subsets, provided
that NC is performed in $\mathbb{F}_{256}$, the probability of getting
a linearly independent packet is above $99.6\%$. The same result
is obtained when $4$ packets are obtained from each of the $k=25$
first subsets. The constraints introduced on the subsets do not degrade
significantly the generation recovery performance compared to plain
NC. This result is mainly obtained due to the fact that one considers
packets received from the first (largest) subsets.
\end{exmple}

\begin{table}[H]
\begin{centering}
\begin{tabular}{|c|c|c|c|}
\hline 
$k$ & $\ell$ & $\mathbb{F}_{2}$ & $\mathbb{F}_{256}$\tabularnewline
\hline 
\hline 
$50$ & $2$ & $0.71$ & $0.0039$\tabularnewline
\hline 
$25$ & $4$ & $0.71$ & $0.0039$\tabularnewline
\hline 
$33$ & $3$ & $0.42$ & $1.53\times10^{-5}$\tabularnewline
\hline 
$49$ & $2$ & $0.23$ & $5.98\times10^{-8}$\tabularnewline
\hline 
$48$ & $2$ & $0.06$ & $9.13\times10^{-13}$\tabularnewline
\hline 
$32$ & $3$ & $0.06$ & $9.13\times10^{-13}$\tabularnewline
\hline 
$24$ & $4$ & $0.06$ & $9.13\times10^{-13}$\tabularnewline
\hline 
$47$ & $2$ & $0.015$ & $1.39\times10^{-17}$\tabularnewline
\hline 
$45$ & $2$ & $0.00097$ & $3.24\times10^{-27}$\tabularnewline
\hline 
\end{tabular}
\par\end{centering}
\centering{}\caption{\label{tab:probability-k-sets}Probability $P_{\text{F}}\left(\ell,k\right)$
of getting $\ell$ linearly dependent vectors chosen at random from
consecutive subsets $\mathcal{A}_{1}$ to $\mathcal{A}_{k}$}
\end{table}

To prove Property~\ref{eq:LinComb} for $k>i$, consider an intermediate
node that received two linearly independent segments $a_{i}^{1}\in\mathcal{A}_{i}$
and $a_{i}^{2}\in\mathcal{A}_{i}$. The $i-1$ first entries of $a_{i}^{1}$
and $a_{i}^{2}$ are zero, and their $i$-th entries $a_{i,i}^{1}$
and $a_{i,i}^{2}$ are non-zero. Then, as $\mathbb{F}_{q}$ is a group
for multiplication, considering any $\alpha_{1}\in\mathbb{F}_{q}^{*}$,
there exists $\alpha_{2}\in\mathbb{F}_{q}^{*}$ such that $\alpha_{1}a_{i,i}^{1}+\alpha_{2}a_{i,i}^{2}=0$.
Moreover, since $a_{i}^{1}$ and $a_{i}^{2}$ are linearly independent,
one has $b=\alpha_{1}a_{i}^{1}+\alpha_{2}a_{i}^{2}\neq0$. Let $k$
be the smallest index such that $b_{k}\neq0$. Necessarily $k>i$
and $b\in\mathcal{A}_{k}$.

\newtheorem{rem}{Remark}
\begin{rem}
Imposing an ordering in the subsets $\mathcal{A}_{k}$ might look
inefficient. Nevertheless, due to pipelining behavior, this is not
a problem. When there is a single path between a client and a source,
Property~\ref{eq:LinIndep} ensures that all contents are innovative.
If $\ell$ distinct paths connect the client to one or several sources,
the first interest packet in the pipeline should bring back $\ell$
linearly independent data packets thanks to Property~\ref{eq:Rank}.
Then again, thanks to Property~\ref{eq:Rank}, the $k$ first pipelined
interest packets are likely to bring back $k\ell$ linearly independent
data packets, see Table~\ref{tab:probability-k-sets}. Consequently,
when $\ell$ distinct paths connect the client to one or several sources,
it is unlikely that this client will need to send interests for contents
in the subsets $\mathcal{A}_{k}$ with $k$ close to $n$. This opens
the potential for an optimization of the size of the pipeline.
\end{rem}

\section{MICN protocol}

\label{sec:Protocol}

This section describes the MICN protocol, focusing on the interest
and content processing using the MILIC construction presented in Section~\ref{subsec:MILIC}
to recover linearly independent content with each interest in the
context of ICN.

\subsection{Content Segmentation and Naming}

\label{sec:Content-SegName}

The original content $C$ is partitioned into $G$ smaller groups
of segments, called \emph{generations }$C=\left[c_{1},c_{2},...,c_{G}\right]$.
Each generation $c_{g}$, $g=1,\dots,G$ contains $n$ \emph{equally-sized}
segments 
\[
c_{g}=\left[c_{g,1},c_{g,2},...,c_{g,n}\right].
\]
The NC operations are restricted to segments that belong to the same
generation and assumed to be performed in $\mathbb{F}_{q}$. 

A MILIC-compliant coded segment, whose encoding vector in the subset
$\mathcal{A}_{i}$, $i=1,\dots,n$, is defined as
\[
\widetilde{c}_{g,i}=\sum_{j=i}^{n}a_{j}c_{g,j}
\]
with $a_{i}\neq0$. The entries of $c_{g,j}$, $j=1,\dots,n$ and
$\widetilde{c}_{g,i}$ are represented as elements of $\mathbb{F}_{q}$.
Any coded segment $\widetilde{c}_{g,i}$ is identified by a prefix,
a generation ID $g$, a MILIC index $i$, and the encoding vector
$a=\left(0,\dots,0,a_{i},\dots,a_{n}\right)\in\mathbb{\ensuremath{F}}_{q}^{n}$
to indicate the weight of each original segment in $\widetilde{c}_{g,i}$.
Consequently, we propose to identify $\widetilde{c}_{g,i}$ by the NDN name
\texttt{\small{}<prefix>/micn/<}{\small{}$g$}\texttt{\small{}>/<}{\small{}$i$}\texttt{\small{}>/<}{\small{}$a_{i},\dots,a_{n}$}\texttt{\small{}>}
(\texttt{micn} indicates that the content is network coded). Other
naming conventions are possible with MICN.

\subsection{Requesting MILIC-compliant contents}

\label{sec:Requesting-MILIC-content}

According to the naming convention of data packets, see Section~\ref{sec:Content-SegName},
the name carried by the interest $I_{g,i}$ for a coded content from
$C$ belonging to the generation $g$ and with an encoding vector
in $\mathcal{A}_{i}$ is\texttt{\small{} <prefix>/micn/<}{\small{}$g$}\texttt{\small{}>/<}{\small{}$i$}\texttt{\small{}>}. 

Contrary to other proposals integrating NC to ICN, this interest format
allows the client nodes to pipeline multiple interests for the same
generation, provided that different $\mathcal{A}_{i}$ are specified
in the names. 

In practice, a client starts sending successive interests for contents
in a given generation $g$, starting from packets with encoding vectors
in $\mathcal{A}_{1},\mathcal{A}_{2},\dots,\mathcal{A}_{\rho}$, where
$\rho$ is the pipeline size. Additional interests are sent once the
content starts flowing back. The pipeline size $\rho$ limits the
number of outstanding interests from a client node at any time.

Each interest $I_{g,i}$ has an associated time-out. If no innovative
content in response to $I_{g,i}$ is received before time out, the
interest $I_{g,i}$ is sent again. Time out may occur, \emph{e.g.},
in case of losses of the interest or data packets.

\subsection{MICN-compliant PIT}

\label{sec:MICN-compliant-PIT}

Compared to the classical NDN PIT, a MICN-compliant PIT identifies
interests requesting coded segments with the same prefix and generation
ID as \emph{related interests.} PIT entries for related interests
are grouped in a sub-table (identified by prefix and generation ID
$g$). Each entry itself includes the associated index $i$, \emph{nonce
$\nu$,} as well as the \emph{in} and \emph{out} faces. The PIT entries
are sorted by order of arrival.

Figure~\ref{fig:NC-compliant-PIT} illustrates a part of a MICN-compliant
PIT at a given node with three faces $f_{1},f_{2}$, and $f_{3}$.
Three interests have been received and have been forwarded. The two
interests associated with $\mathcal{A}_{1}$ are considered as different
since they have different nonce, which implies that different clients
sent them.


\subsection{Just-in-Time Content Re-encoding / Replying}

\label{sec:Interest-Replying}

In plain NDN, whenever a node can satisfy an interest, a copy of the
requested content is sent immediately. In MICN, as in some other NC-NDN
protocols, nodes do not just forward a copy of the cached coded content
as an answer to the matching interests. They linearly combine cached
contents from the same generation to generate a new coded segment.
\begin{figure}[htpb]
\centering
\begin{tabular}{|c|c|c|c|}
\hline 
\multicolumn{4}{|c|}{PIT}\tabularnewline
\hline 
\multicolumn{4}{|c|}{\texttt{\small{}<prefix>/micn/<}{\small{}$g$}\texttt{\small{}>}}\tabularnewline
\hline 
\hline 
Index & Nonce & in-faces & out-faces\tabularnewline
\hline 
1 & $\nu_{1}$ & $f_{1}$ & $f_{2},f_{3}$\tabularnewline
\hline 
1 & $\nu_{2}$ & $f_{1}$ & $f_{2},f_{3}$\tabularnewline
\hline 
$2$ & $\nu_{3}$ & $f_{2}$ & $f_{1},f_{3}$\tabularnewline
\hline 
$\vdots$ & $\vdots$ & $\vdots$ & $\vdots$\tabularnewline
\hline 
\multicolumn{4}{|c|}{\texttt{\small{}<prefix>/micn/<}{\small{}$g$}\texttt{\small{}'>}}\tabularnewline
\hline 
 &  &  & \tabularnewline
\hline 
\end{tabular}
\caption{\label{fig:NC-compliant-PIT}MICN compliant PIT}
\end{figure}

\begin{figure}[htpb]
\centering
\begin{tabular}{|c|c|c|c|}
\hline 
\multicolumn{4}{|c|}{PIT}\tabularnewline
\hline 
\multicolumn{4}{|c|}{\texttt{\small{}<prefix>/micn/<$g$>}}\tabularnewline
\hline 
\hline 
Index & Nonce & in-faces & out-faces\tabularnewline
\hline 
1 & $\nu_{1}$ & $f_{1}$ & $f_{2},f_{3}$\tabularnewline
\hline 
1 & $\nu_{2}$ & $f_{1}$ & $f_{2},f_{3}$\tabularnewline
\hline 
$2$ & $\nu_{3}$ & $f_{2}$ & $f_{1},f_{3}$\tabularnewline
\hline 
$\vdots$ & $\vdots$ & $\vdots$ & $\vdots$\tabularnewline
\hline 
$k$ & $\nu_{k}$ & $f_{1}$ & \tabularnewline
\hline 
\multicolumn{4}{|c|}{\texttt{\small{}<prefix>/micn/<}{\small{}$g$}\texttt{\small{}'>}}\tabularnewline
\hline 
 &  &  & \tabularnewline
\hline 
\end{tabular}
\caption{\label{fig:NC-compliant-PIT-1}MICN compliant PIT: the interest with
index $k$ lead to cache hit and is temporarily stored in the PIT
until the queue of face $f_{1}$ is empty to send back the associated
data packet}
\end{figure}

In MICN, the reply strategy is further modified, compared to plain
NDN. A node waits until the queue of a face is empty before generating
a coded segment that satisfies a pending interest on this face. This
allows the node to use its latest cached contents when replying, hence
sending more diverse content through the network. To achieve this
one packet queue is considered at the faces that is filled only when
the packet in transit is completely delivered. The process to achieve
this just-in-time re-encoding is detailed in Sections~\ref{sec:CS-lookup}
and \ref{sec:Content-Processing}.

\subsection{Interest processing}

\label{sec:Interest-processing}

When a node receives an interest $I_{g,i}$, it initially performs
loop detection. If an interest with the same nonce has already been
received, $I_{g,i}$ is considered as a looping interest. Otherwise,
the node processes the interest. It can either reply using a content
generated from its CS or further forwards the interest to the network.

\subsubsection{CS lookup}

\label{sec:CS-lookup}

Like the PIT, the \emph{related contents} (\emph{i.e.}, contents with
the same prefix and generation ID) are grouped in the CS. The \emph{CS
lookup }starts by identifying the related content matching the received
interest. A \emph{cache hit} occurs when this cached content can be
used to generate a coded segment belonging to the subset requested
in the interest.

In case of a cache hit, the node schedules a reply for the interest.
The node first checks the outgoing queue of the face where the interest
arrived. If the queue is empty, the content is immediately sent in
a data packet. Otherwise, a reply is generated only when the queue
becomes empty. In our implementation, this scheduling is achieved
by creating a \emph{volatile} PIT entry to store the incoming face,
nonce, \emph{etc}., but without specifying an outgoing face, since
the interest does not require to be forwarded. See, for example, the
interest with index $k$ in Figure~\ref{fig:NC-compliant-PIT-1}.
\begin{center}
\begin{figure}[H]
\centering{}\subfloat[\label{fig:reply-MICN}Immediate re-encoding: $\widetilde{c}_{3}=\alpha_{1}c_{1}+\alpha{}_{3}c_{3}$
is put in the outgoing queue of face $f_{1}$ before the reception
and processing of content packet $c_{2}.$]{\includegraphics[width=1\columnwidth]{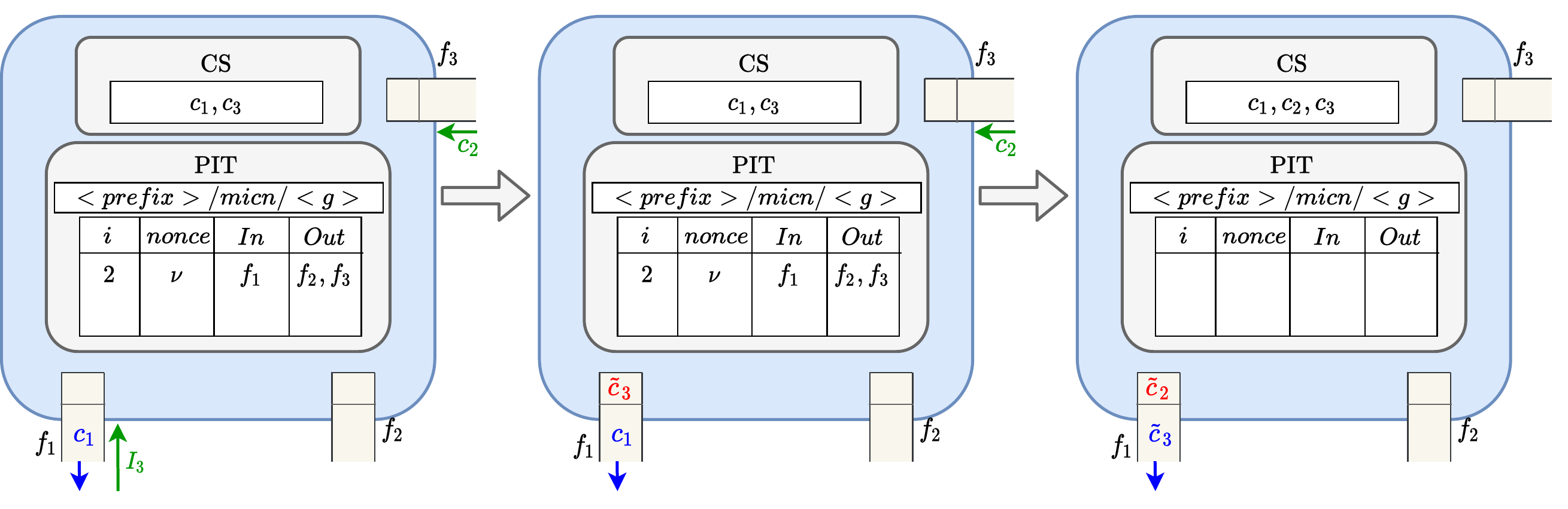}}\linebreak{}
\subfloat[\label{fig:reply-MICN-optm}Just-in-time re-encoding with MICN: $\widetilde{c}_{3}=\alpha_{1}c_{1}+\alpha_{2}c_{2}+\alpha{}_{3}c_{3}$
is put in the outgoing queue of face $f_{1}$ only once this queue
is empty; this gives the opportunity to the later received $c_{2}$
on face $f_{2}$ to be included in $\widetilde{c}_{3}$.]{\includegraphics[width=1\columnwidth]{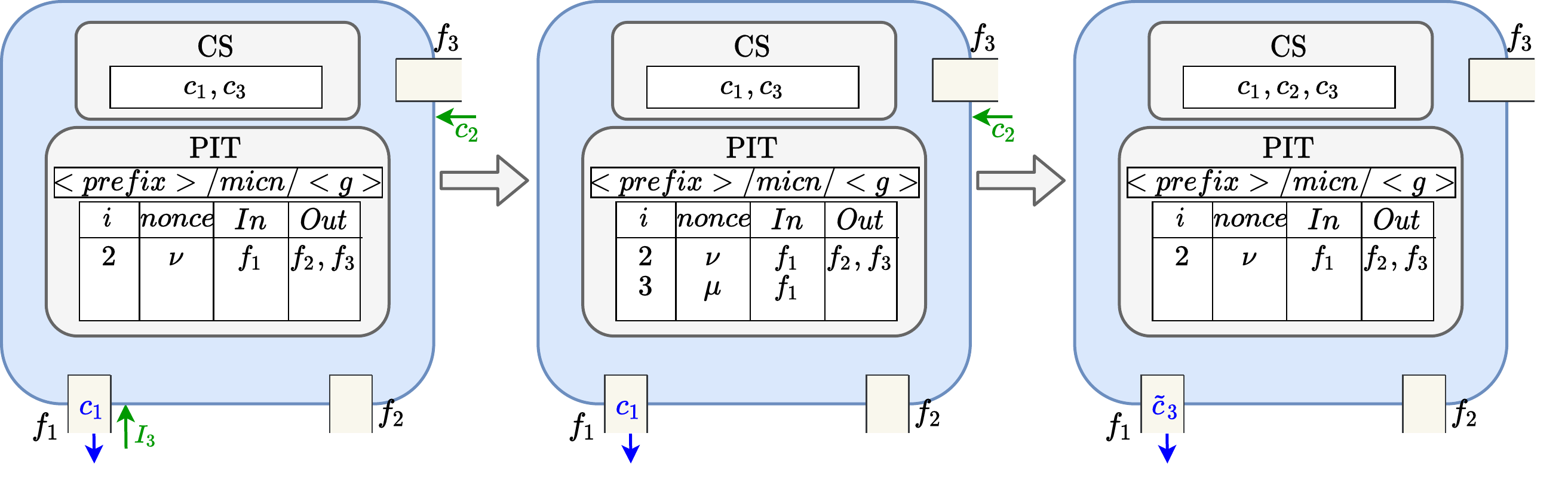}}\caption{\label{fig:Interest-Replying}Re-encoding cached content}
\end{figure}
\par\end{center}

Fig.~\ref{fig:reply-MICN} illustrates the state of a node that has
enough content in its CS to respond to the incoming interest $I_{3}$,
it immediately uses the cached related content to generate a response
$\widetilde{c}_{3}$, but the content remains in the queue until the
content $c_{1}$ is transmitted. While the node in Fig.~\ref{fig:reply-MICN-optm}
waits until $c_{1}$ is transmitted since it may receive more content
and have a more diverse CS (since more content from the same generation
is requested). So a volatile PIT entry is generated that is replied
as soon as the queue becomes empty. 

\subsubsection{Interest Forwarding}

\label{sec:interest-forwarding}

In case of a cache miss, the node forwards the interest to its next
hop neighbors on the faces in the FIB (except the incoming face) and
creates a PIT entry, which records the incoming and outgoing faces.
Unlike classical NDN, different nonces result in different entries,
see Figure~\ref{fig:NC-compliant-PIT}.

According to the management of the FIB, multiple interest forwarding
strategies can be implemented depending on the subset of chosen faces
to forward the content. In this paper, to take advantage of the multiple
paths to the source(s) and to have the opportunity to receive multiple
linearly independent segments, the FIB is filled with all faces that
can lead to a source \emph{without looping back} to the node. Then,
the \emph{multicast forwarding strategy }where the interests are forwarded
on all faces in the FIB is used, as suggested in \cite[Section 5.2.2]{Afanasyev2018}.

\subsection{Content Processing}

\label{sec:Content-Processing}

When a coded segment arrives at a node from one of its faces, it adds
it in its cache if it is linearly independent with the already cached
related contents. The updated cache might then satisfy some additional/new
interests.

The node uses its updated cache to reply to the pending interests.
Whenever the queue of a face f is empty, the node checks if any pending
interest on that face can be satisfied utilizing the current state
of the cache. It answers the oldest PIT entry, that may be satisfied
and removes the entry.

\subsection{Optimizations}

\label{sec:Optimizations}

In this section, we introduce some optimization compared to the classical
NDN to improve the performance of MICN in an NC-NDN scenario.

\subsubsection{Content Redirection}

\label{sec:Content-Redirection}

A node can receive an interest on an alternate face while the same
interest (same nonce) is still pending at the first face due to the
Just-in-Time content re-encoding of MICN, see Section~\ref{sec:Interest-Replying}.
Such interest brings information that there exists an alternate path
to the client. If the queue associated with this alternate/second
face is empty and the node has matching content, the content is immediately
redirected to the client via this alternate face. This redirection
is likely to improve the network utilization, by benefiting from all
paths leading to the client.

Fig.~\ref{fig:Content-Redirection} depicts the state of a node that
receives interest $I_{3}$ with the same nonce $\nu$ from an alternate
face $f_{2}$ with an empty queue. Since the node has enough content
to generate a reply for the interest but the face $f_{1}$ is busy,
the node redirects the content via the alternate face to send the
reply immediately and possibly benefit from a second path to the client. 
\begin{center}
\begin{figure}[H]
\begin{centering}
\includegraphics[width=1\columnwidth]{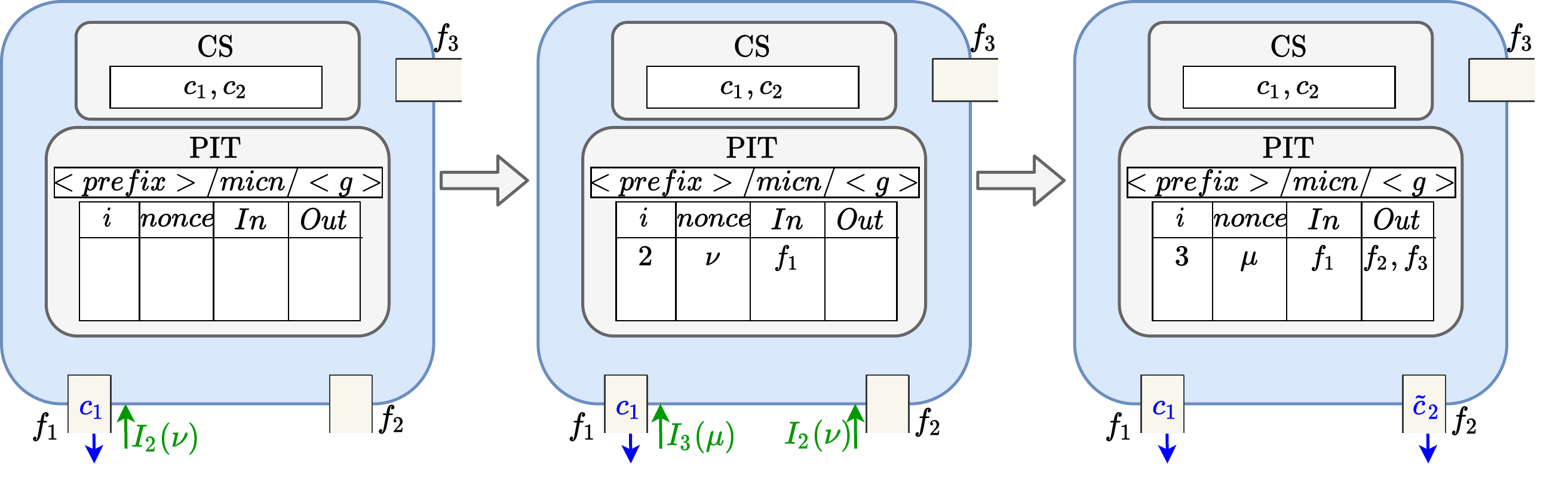}
\par\end{centering}
\caption{\label{fig:Content-Redirection}Content redirection on face $f_{2}$:
during the transmission of $c_{1}$, an interest for content associated
to $\mathcal{A}_{2}$ has been received from face $f_{1}$ (left)
and then from face $f_{2}$ (middle); since the outgoing queue of
face $f_{1}$ is still occupied, $\widetilde{c}_{2}$ is transmitted
on face $f_{2}$ (right).}
\end{figure}
\par\end{center}

\subsubsection{Interest Cancellation (MICN-IC)}
\label{sec:Interest-Cancellation}

\begin{figure}[H]
\centering
\includegraphics[width=\textwidth]{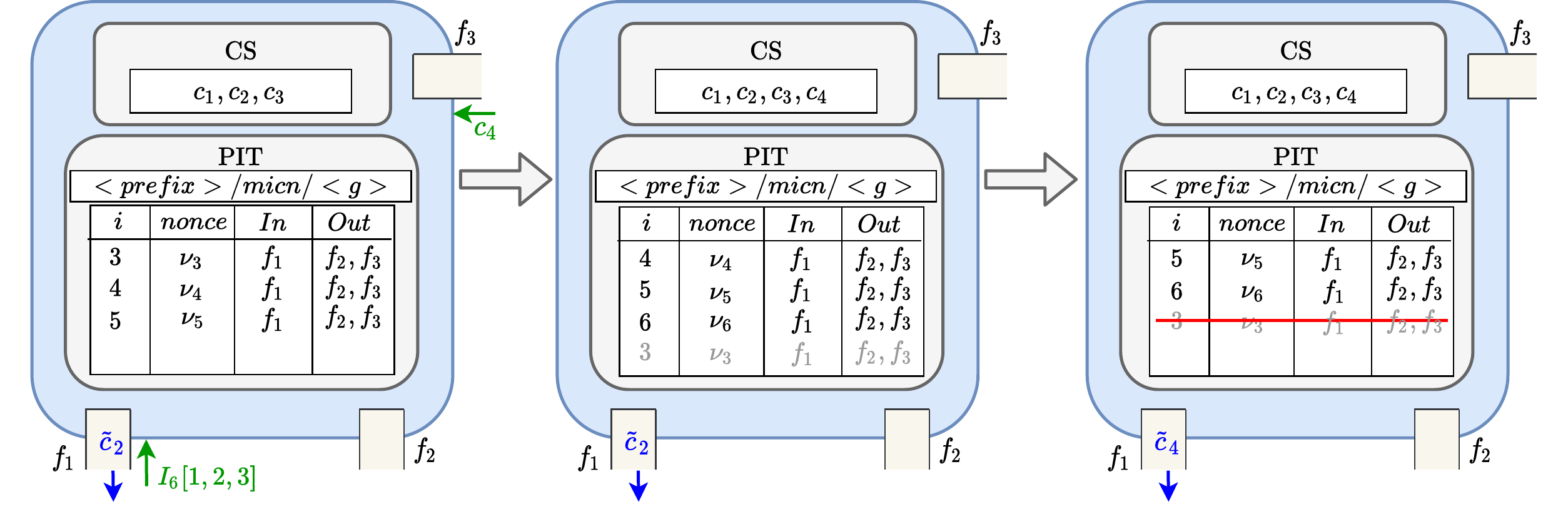}\caption{\label{fig:Cancellation}Interest Cancellation: An interest for packets
associated to $\mathcal{A}_{6}$ is coming from face $f_{1}$ indicating
that the source has already access to content associated to $\mathcal{A}_{1}$,
$\mathcal{A}_{2}$, and $\mathcal{A}_{3}$ (left); The pending interest
for content associated to $\mathcal{A}_{3}$ is first tagged with
low priority (middle); This pending interest is canceled as soon as
an interest associated to a subset of higher index (here $\mathcal{A}_{4}$)
is replied to (right).}
\end{figure}

We observe that the content continues to flow in the network due to
delay differences in different parts of the network even after the
client nodes have received a decodable generation. In order to reduce
the traffic represented by redundant contents, we introduce the concept
of \emph{interest cancellation}.

For that purpose, client nodes add information about the content they
already received. To achieve this, the optional \emph{client identifier}
and \emph{state} fields are added in the interest packets. The client
identifier field is a hash of the client node identifier, while the
state field bears the information of subsets as defined by MILIC for
which that client has already available content. Such content may
have been directly obtained or deduced after Gaussian elimination
involving several data packets. The state field may, \emph{e.g.},
be represented by a bitmap indicating the available indices. 

A node, when receiving interest with the state of a client, may ignore
the pending interests referencing subsets $\mathcal{A}_{i}$ for the
indices $i$ for which content is already available. Nevertheless,
this node does not immediately delete them. Instead, they get low
priority for replies, contrary to other interests in the PIT, which
have a normal priority. Keeping and answering these low-priority interests
may still be useful, according to Properties~\ref{eq:Rank} and \ref{eq:LinComb}:
NC contents sent as replies even for subsets from which content is
already available may bring information with a \emph{high probability}.

A reply to interest with a low priority index is sent only if the
outgoing face is empty, and the node cannot generate content as a
reply to interest with a normal priority index. The deletion of low-priority
interests occurs when the node has sent content for an interest with
a higher index to the client. This version of MICN with Interest Cancellation
is referred to as MICN-IC.

Fig.~\ref{fig:Cancellation} illustrates the state of a node that
receives an interest for some content associated to $\mathcal{A}_{6}$.
The interest also carries the state of requesting node indicating
that it already access to contents associated to $\mathcal{A}_{1}$,
$\mathcal{A}_{2}$, and $\mathcal{A}_{3}$. Using this information,
the node sets the pending interest for content associated to $\mathcal{A}_{3}$
to low priority (interest in gray). This low priority interest is
deleted once a content associated to $\mathcal{A}_{4}$ has been sent
to the considered client. 

\section{Evaluation}

\label{sec:Evaluation}

\subsection{Simulation setup}

\label{sec:Simulation-setup}

We implemented our simulator in Python. It includes a generic packet
network simulator (scheduler, link, packet transmission), on top of
which we developed an implementation of the proposed MICN protocol,
and lightweight reimplementations of NDN and NetCodCCN (capturing
the main semantics of these protocols as described in \cite{Afanasyev2018}
and \cite{Saltarin2016}). We experimented protocols over two topologies:
one simple illustrative butterfly topology and a second, more elaborate
topology close to the PlanetLab topology from NetCodCCN \cite{Saltarin2016}.
At the link level, the parameters of our simulations are a propagation
delay of $0.1$ time unit for each packet, a transmission time of
$1$ time unit for data packets, and a very small transmission delay
for interest packets ($1/2^{14}\simeq6\times10^{-5}$). A small amount
of uniformly distributed transmission jitter (between $0$ and $1/2^{18}\simeq3.8\times10^{-6}$)
was also introduced.

In each topology, we consider the following scenario. Several clients
request coded content, divided into generations of 100 segments each.
We study the transmission of one generation. Each source stores a
complete copy of the coded content. We assume that the intermediate
nodes have enough cache space to store all segments of a generation.
All the coding operations are performed in the finite field $\mathbb{F}_{2^{8}}$.

In MICN, the interest pipeline size $\rho=10$ is considered, the
FIB and the interest forwarding are as described in Section~\ref{sec:interest-forwarding}.
At the client, each interest packet has a time out of $10$ time units
(\emph{i.e.,} equivalent the transmission delay of $10$ data packets,
that is a bit longer than the longer round-trip delay). If a client
does not receive innovative content for an interest after this time
interval, it will resend the interest.

The performance is evaluated in terms of download time, \emph{i.e.},
the time it takes for a client to download and decode a generation.
The time needed to perform Gaussian elimination is neglected. An upper
bound of the throughput (content/time unit) received by a client is
given by the maximum flow on the graph from the sources to the node.
From this max-flow, one can derive a lower bound of the download time.
In similar settings, it had been proven that network coding could
approach the max-flow bound \cite{Ho2006}, hence representing a meaningful
benchmark. Another metric of interest is the total number of data
packets exchanged in the network until all clients have retrieved
the generation with no interest or data packet is present in the network
anymore.

\subsection{Results with the butterfly topology}

We first analyze the behavior of MICN on a simple butterfly topology
with two sources $S_{1}$ and $S_{2}$ and two clients $U_{1}$ and
$U_{2}$ connected through a set of intermediate caching routers as
represented in Fig.~\ref{fig:ButterFly-Topology}.

\begin{figure}[H]
\centering{}\includegraphics[width=0.5\columnwidth]{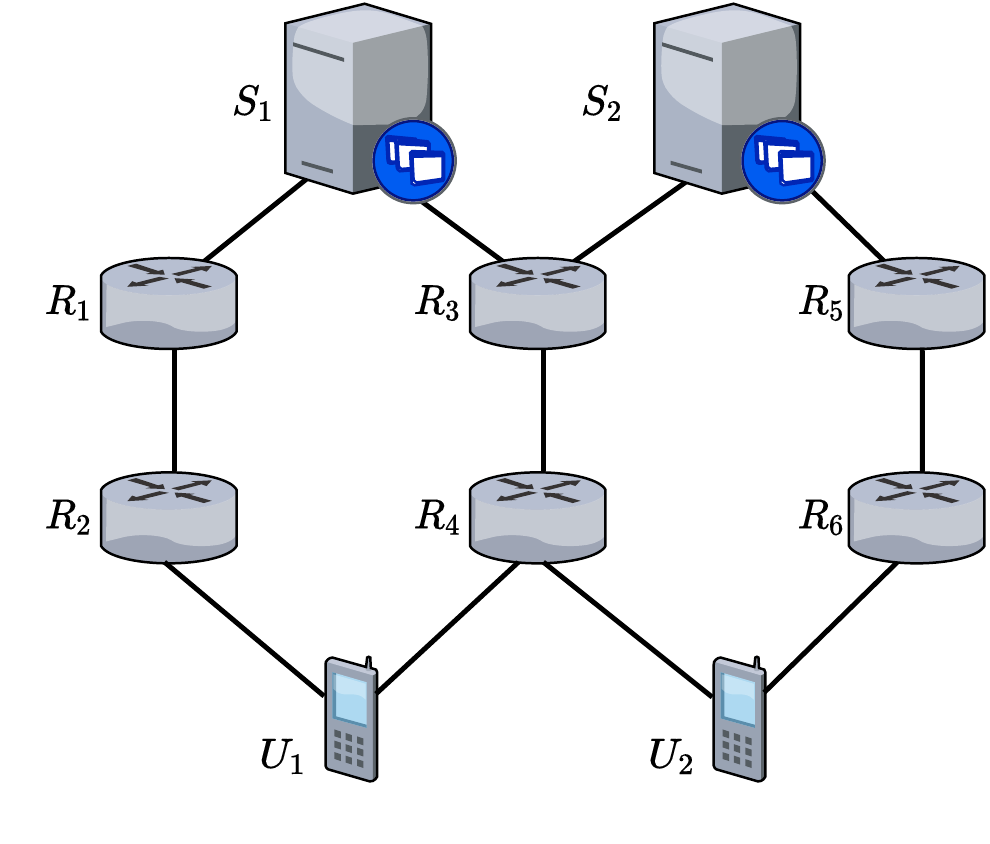}\caption{\label{fig:ButterFly-Topology}Butterfly topology}
\end{figure}

The performance of MICN mainly depends on how the bottleneck link
($R_{3}\leftrightarrow R_{4}$) is used. With classical NDN, the two
clients $U_{1}$ and $U_{2}$ should request precisely the same segments
on the middle link to improve performance. Nevertheless, the clients
would require topology knowledge and coordination to do so. However,
with NC, this is not required, and the clients can simultaneously
send their interests to all their available faces.
\begin{flushleft}
\begin{figure}[H]
\begin{centering}
\subfloat[NDN]{\begin{centering}
\includegraphics[width=0.5\columnwidth]{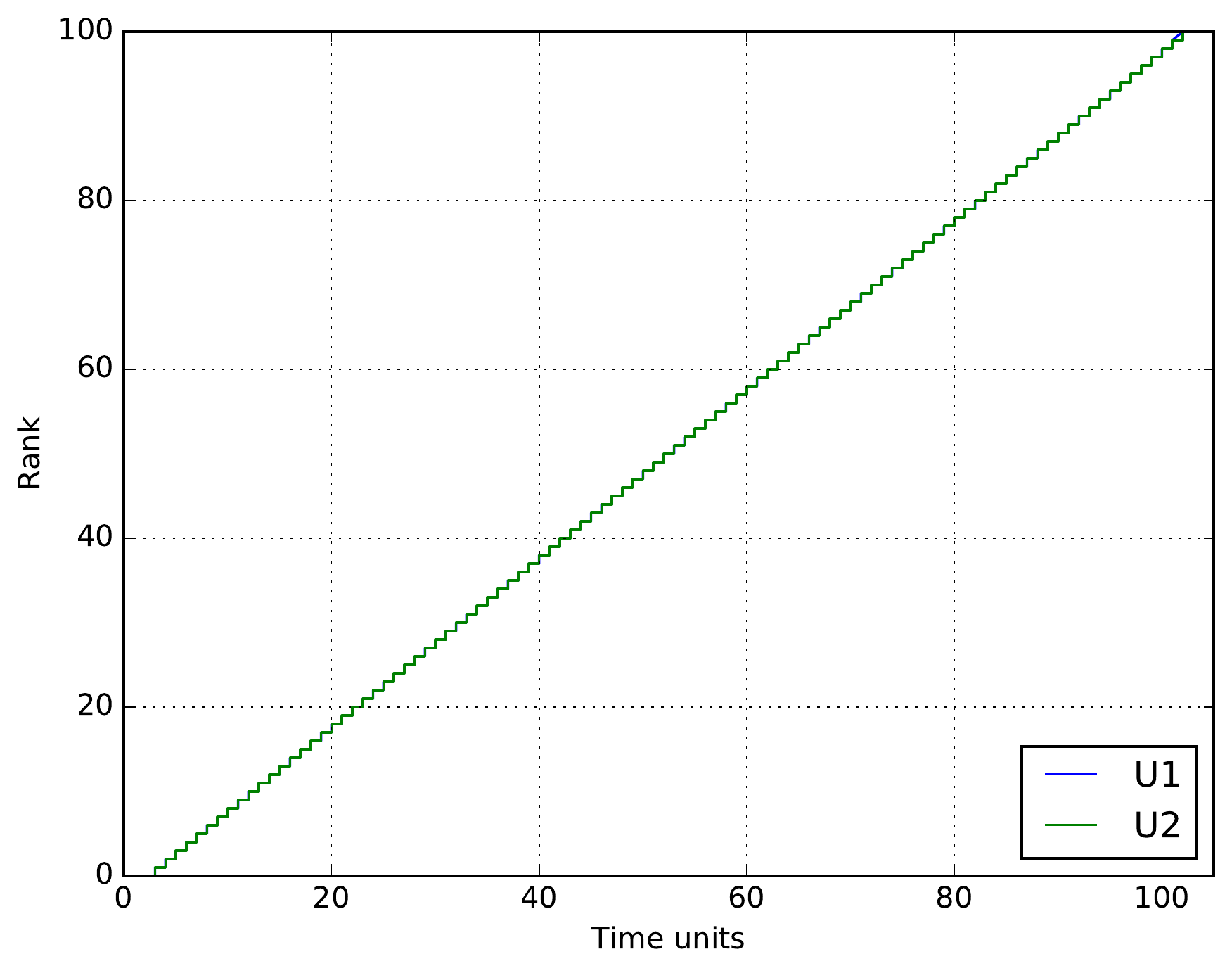}
\par\end{centering}
}\subfloat[NetCodCCN]{\centering{}\includegraphics[width=0.5\columnwidth]{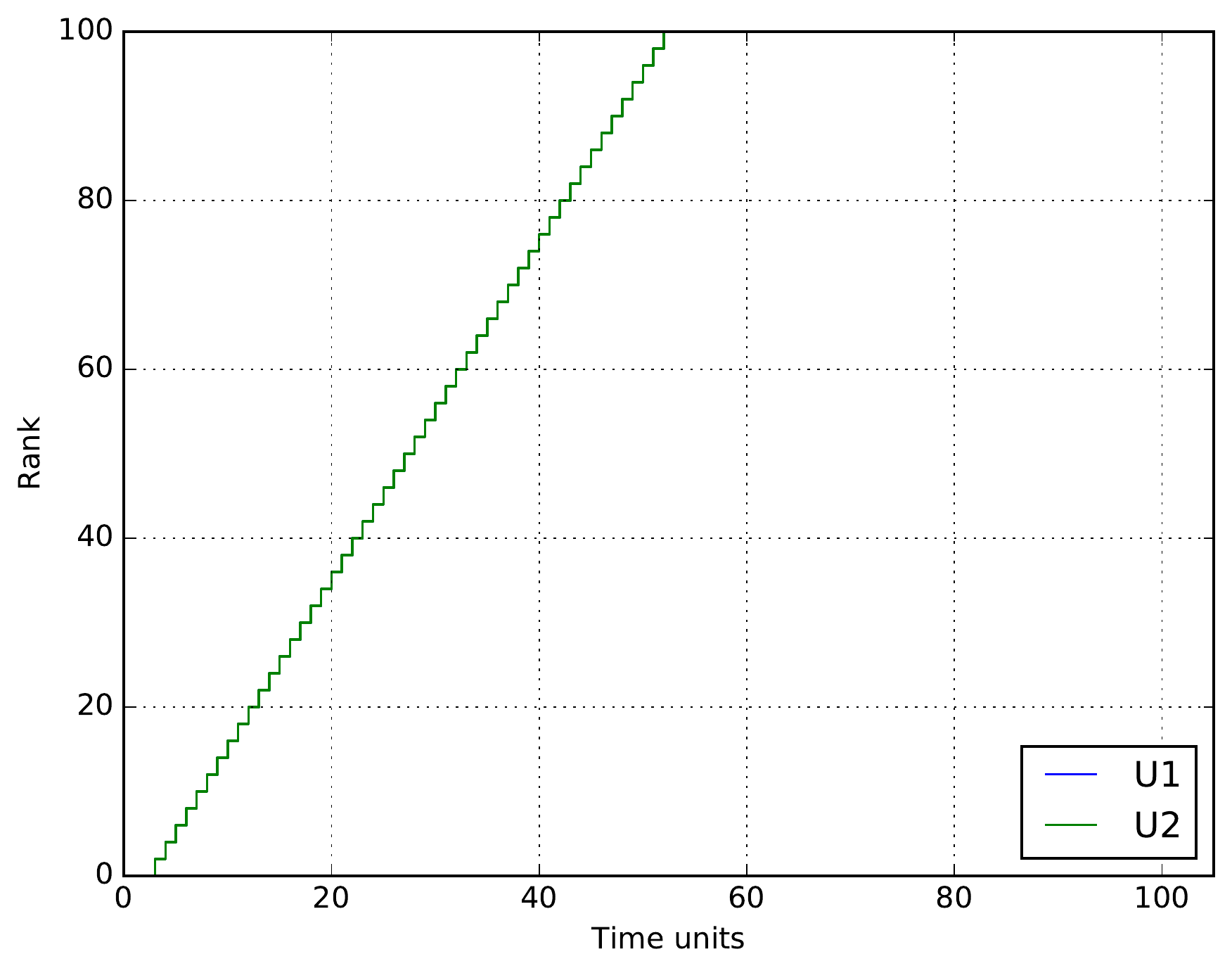}}\\
\subfloat[MICN]{\centering{}\includegraphics[width=0.5\columnwidth]{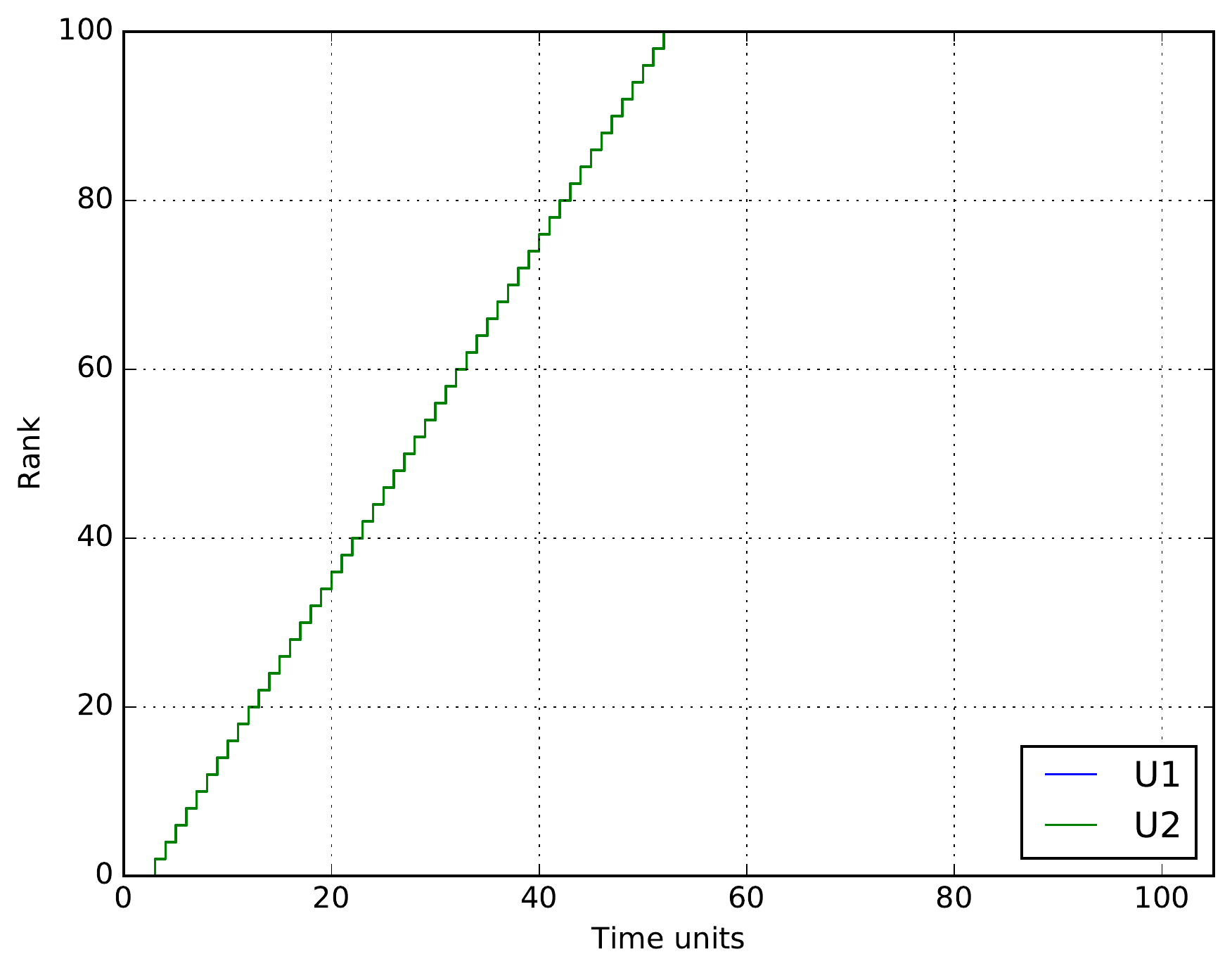}}\subfloat[MICN-IC]{\centering{}\includegraphics[width=0.5\columnwidth]{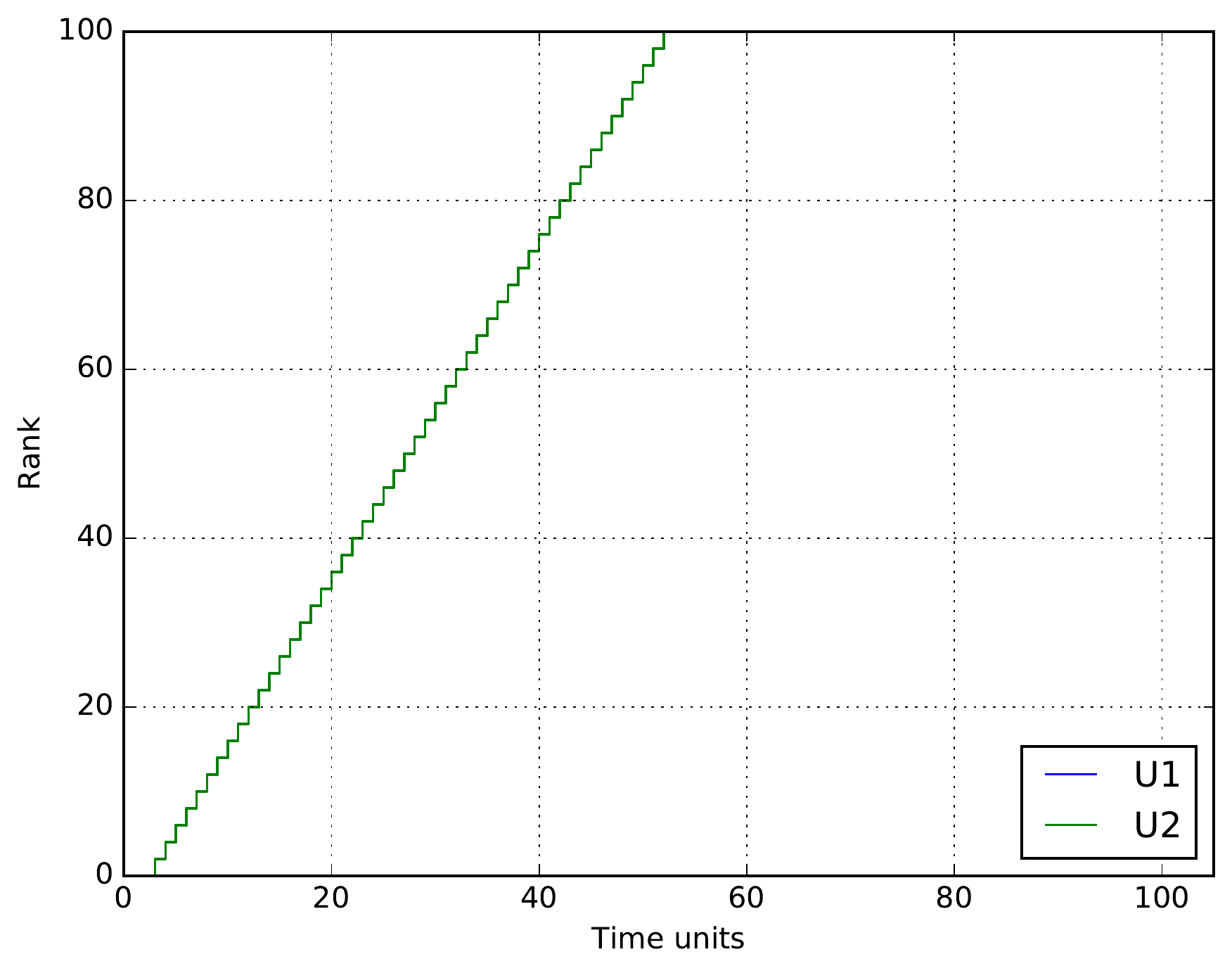}}
\par\end{centering}
\centering{}\caption{Butterfly topology: Rank evolution as a function of time\label{fig:Rank-butterfly}}
\end{figure}
\par\end{flushleft}

Fig.~\ref{fig:Rank-butterfly} shows the rank evolution of the client
nodes over time for MICN, MICN-IC, NetCodCCN, and NDN. The three protocols
retrieve content at both the clients at the max-flow rate, \emph{i.e.},
each data packet received at the client is innovative. After some
initial delay, due to propagation, with all protocols, the clients
receive $2$ linearly independent data packets every $10$ time units.
Nevertheless, there are significant differences in the volume of data
traffic that each protocol generates, as shown in Fig.~\ref{fig:traffic-butterfly}.

\begin{figure}[H]
\begin{centering}
\subfloat[NetCodCCN]{\centering{}\includegraphics[width=0.50\columnwidth]{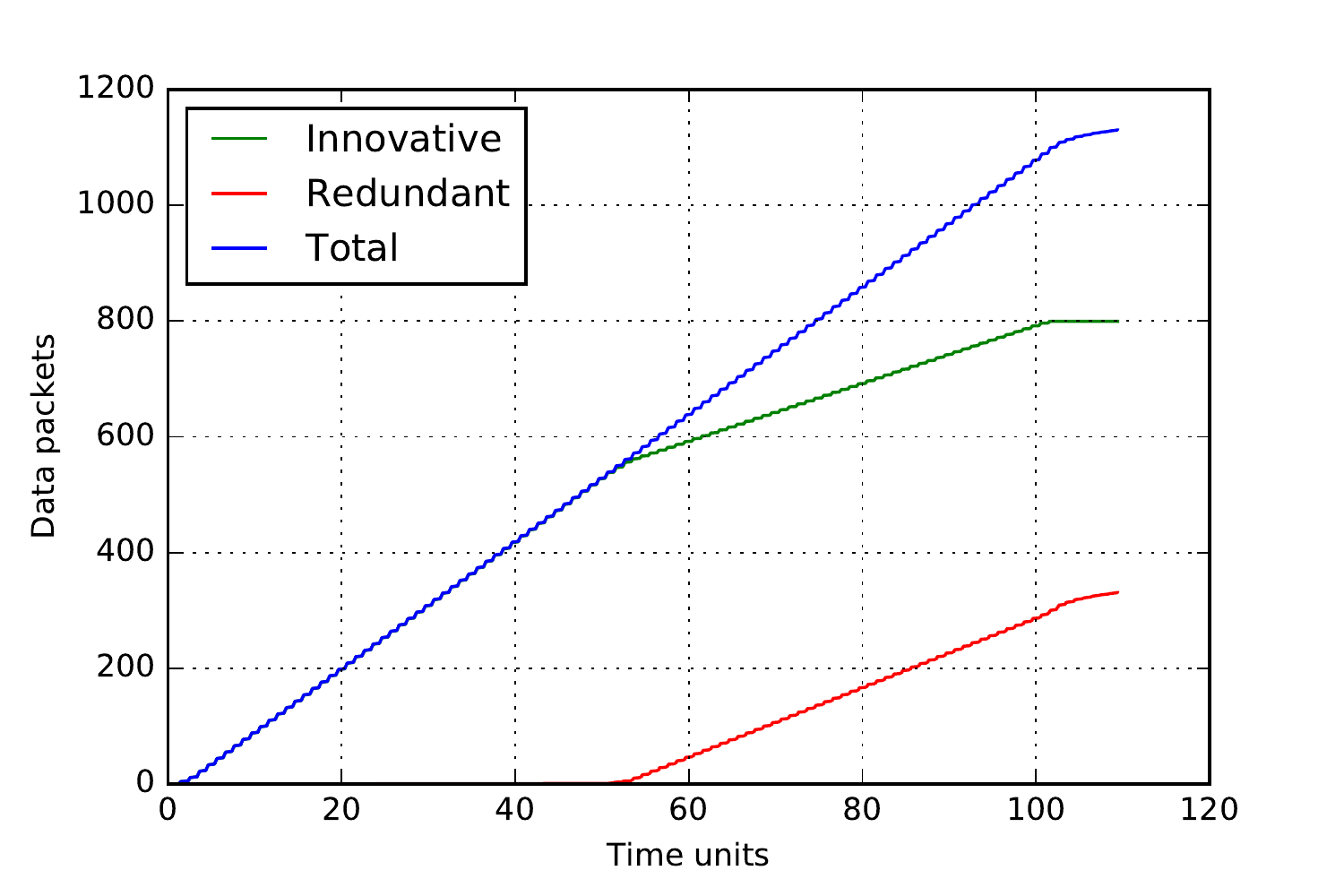}}

\subfloat[MICN]{\centering{}\includegraphics[width=0.50\columnwidth]{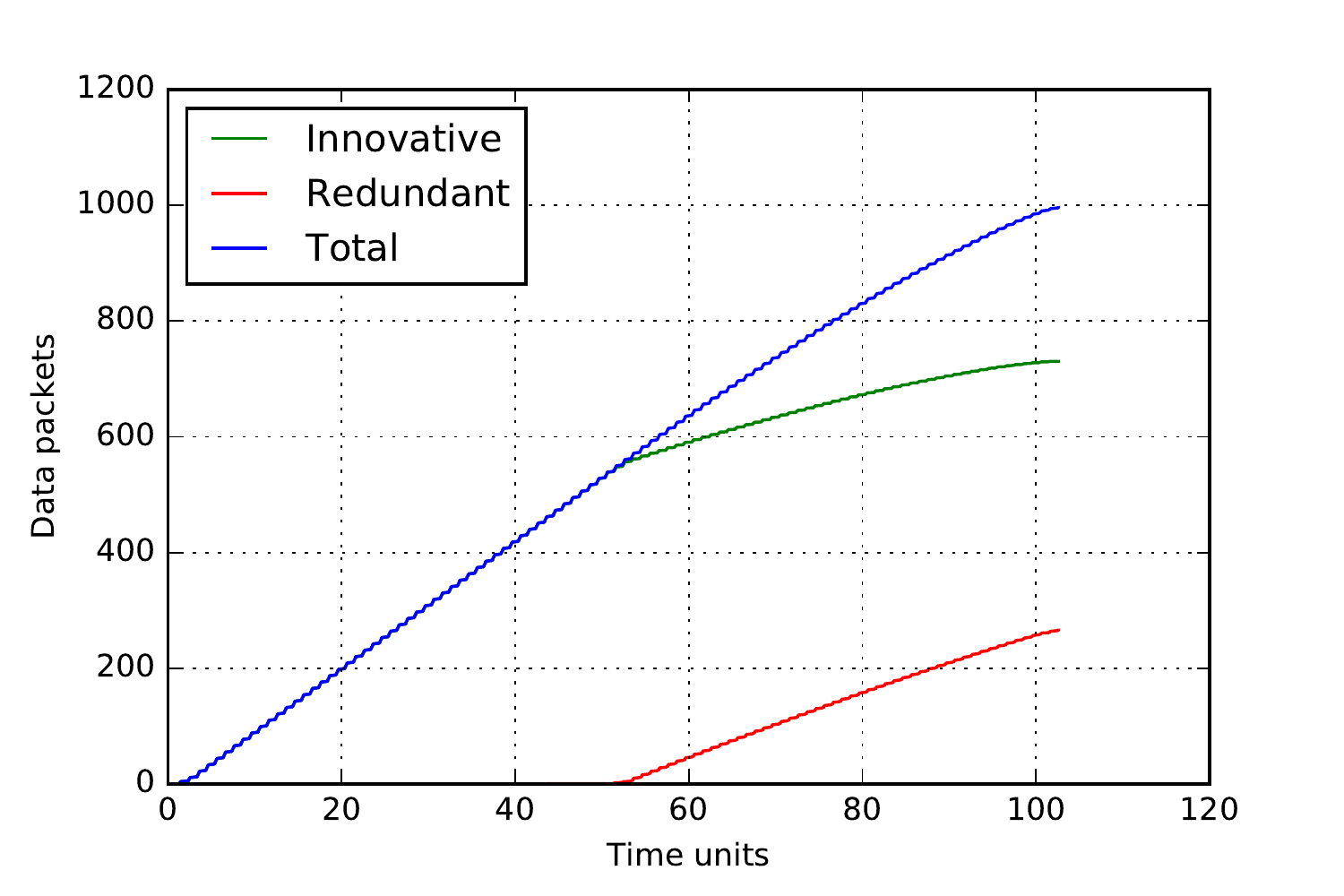}}

\subfloat[MICN-IC]{\centering{}\includegraphics[width=0.50\columnwidth]{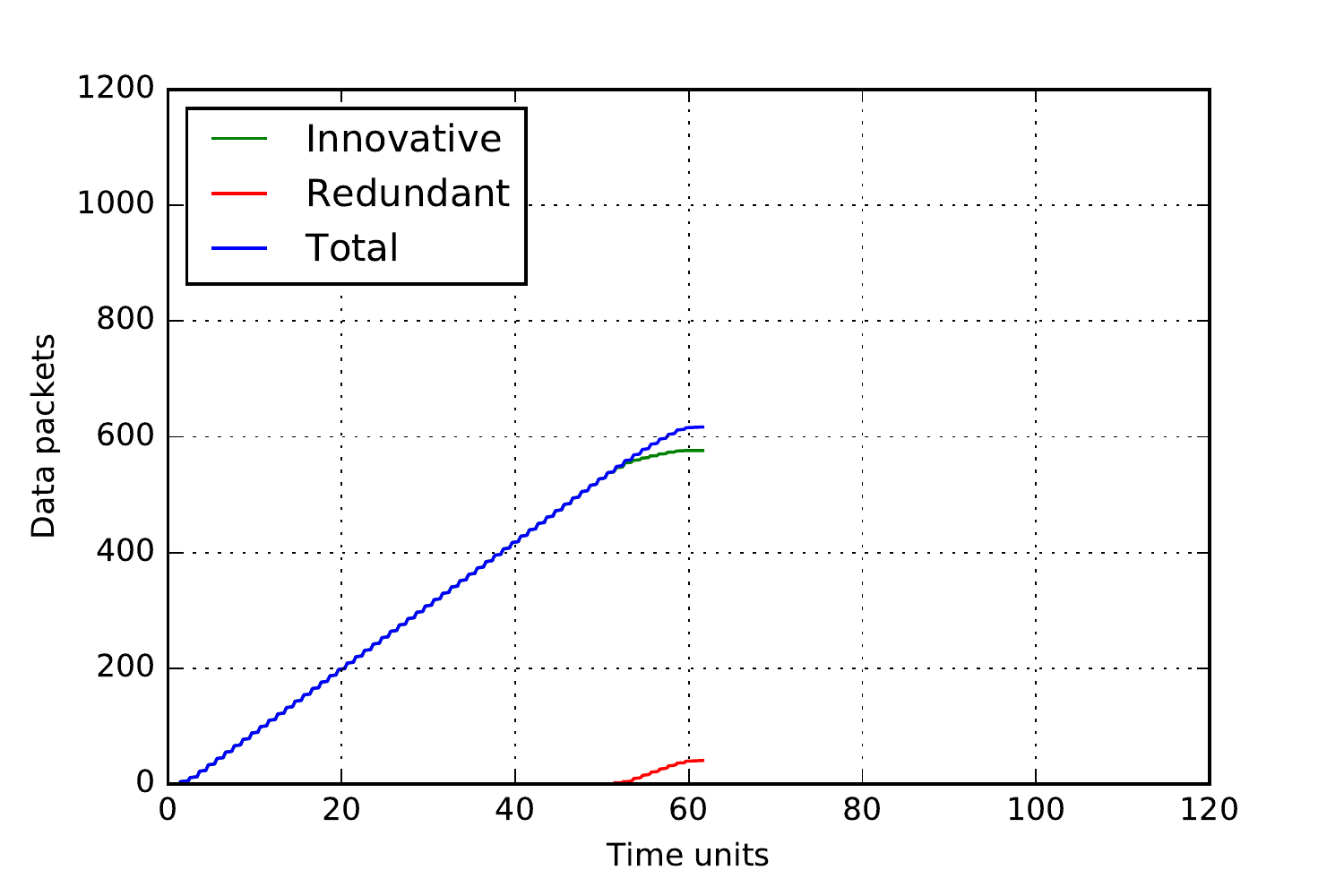}}
\par\end{centering}
\centering{}\caption{Data traffic in the butterfly topology\label{fig:traffic-butterfly}}
\end{figure}

Fig.~\ref{fig:traffic-butterfly} depicts the evolution with time
of the cumulative number of data packets transmitted on all the links
of the network, counted from time $t=0$. The curves end when transmission
of data packets stops. In the beginning, there is only innovative
traffic in the network, \emph{i.e.}, data packets which are innovative
for the routers or the client node receiving them. While towards the
end, even after the clients have received the entire generation, the
outlying interests continue to generate data traffic that is just
redundant. MICN-IC deletes the interests tagged with low priority,
which are pending even if a client has access to content for those
interests. Canceling such interests reduces the redundant data traffic,
at the price of some signaling overhead. Precisely, in the butterfly
topology (Fig.~\ref{fig:ButterFly-Topology}), $10$ transmissions
of data packets over various links are necessary for delivering $2$
data packets to the clients $U_{1}$ and $U_{2}$, \emph{i.e}., $5$
transmissions per packet. For a generation of size $100$, a minimum
of $500$ transmissions are required for both clients to receive the
entire generation. Fig.~\ref{fig:traffic-butterfly} shows that with
MICN-IC, a slightly larger amount of transmissions are required. NetCodCCN
achieves similar download performance, but interests are not canceled
and several data packets are redundant, leading to increased traffic. 

\begin{figure}[H]
\begin{centering}
\includegraphics[width=0.5\columnwidth]{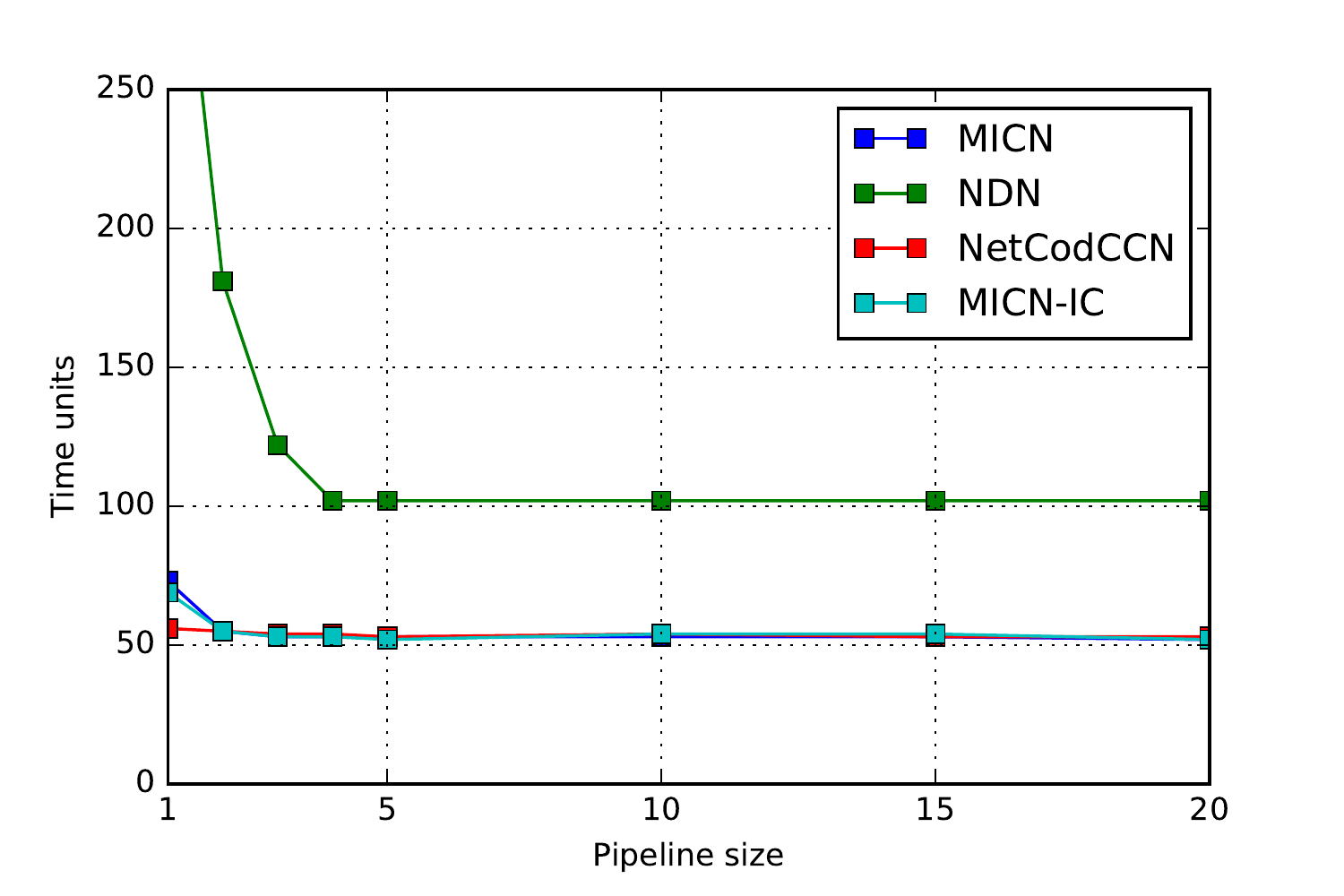}
\par\end{centering}
\caption{\label{fig:pipeline-fly}Butterfly topology: Download time vs pipeline
size, no losses}
\end{figure}

The effect of sending consecutive interests by clients is analyzed
in Fig.~\ref{fig:pipeline-fly}. To have a continuous flow of content
in the butterfly network (in the absence of losses), the clients need
to have at least two outstanding interests at any time (because there
are two paths) and usually even more because of the propagation delays.
In the case of plain NDN with multicast strategy, the link $R_{3}\leftrightarrow R_{4}$
becomes a bottleneck due to no coordination among the clients. Even
when the pipeline size increases, the performance cannot reach the
one obtained with NC. MICN, MICN-IC, and NetCodCCN, however, with
a sufficient pipeline size (here as small as $\rho=5$), can reach
the maximum capacity.

\begin{figure}[H]
\begin{centering}
\includegraphics[width=0.5\columnwidth]{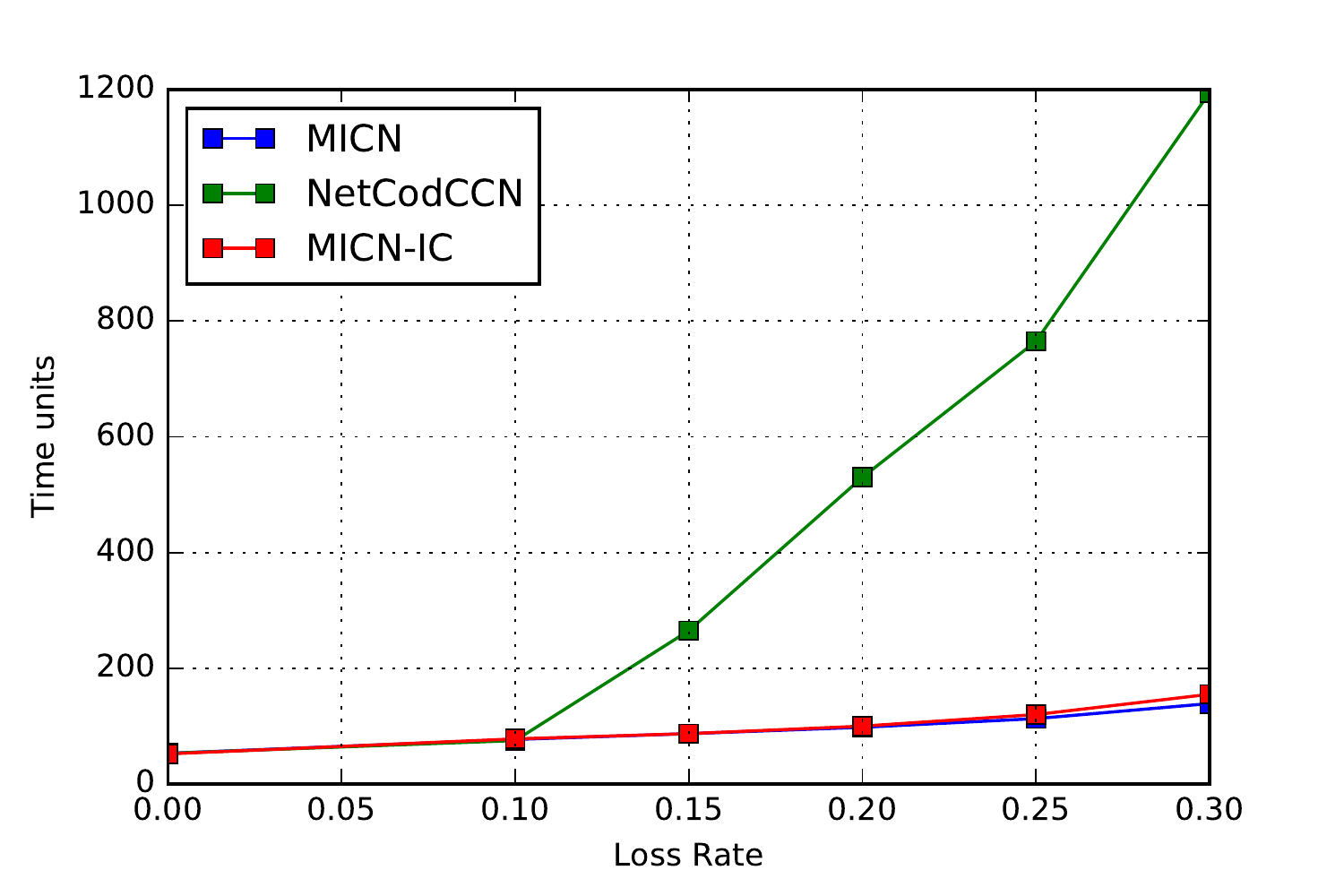}
\par\end{centering}
\caption{\label{fig:losses-fly}Butterfly Topology:Download time vs transmission
error rate}
\end{figure}

Next, we evaluate the performance of MICN in case of losses. Fig.~\ref{fig:losses-fly},
depicts the effect of losses on the performance of the protocols.
We consider transmission losses modeled with a fixed loss probability
for both interest and data packets.\footnote{Notice that NetCodCCN simulations in \cite{Saltarin2016} consider
only segment (data packets) losses, but here both interest and data
packets are prone to losses.} MICN and MICN-IC appear to have much better performance compared
to NetCodCCN. MICN has the advantage of precisely identifying which
interest (pointing to a subset $\mathcal{A}_{i}$) has timed-out (no
matching content received). In NetCodCCN, even if a downstream data
packet is lost, the router will consider the interest satisfied. An
interest repeated due to time-out is considered a new interest, and
the router will typically forward it. In MICN, the repeated interest
will be immediately satisfied by the router's cache.

\subsection{Results with the PlanetLab topology}

The behavior of MICN is then analyzed considering the PlanetLab topology
from \cite{Saltarin2016}, with one source and five client nodes connected
through a set of 20 intermediate caching routers.
\begin{center}
\begin{figure}[H]
\begin{centering}
\subfloat[NetcodCCN]{\centering{}\includegraphics[width=0.50\columnwidth]{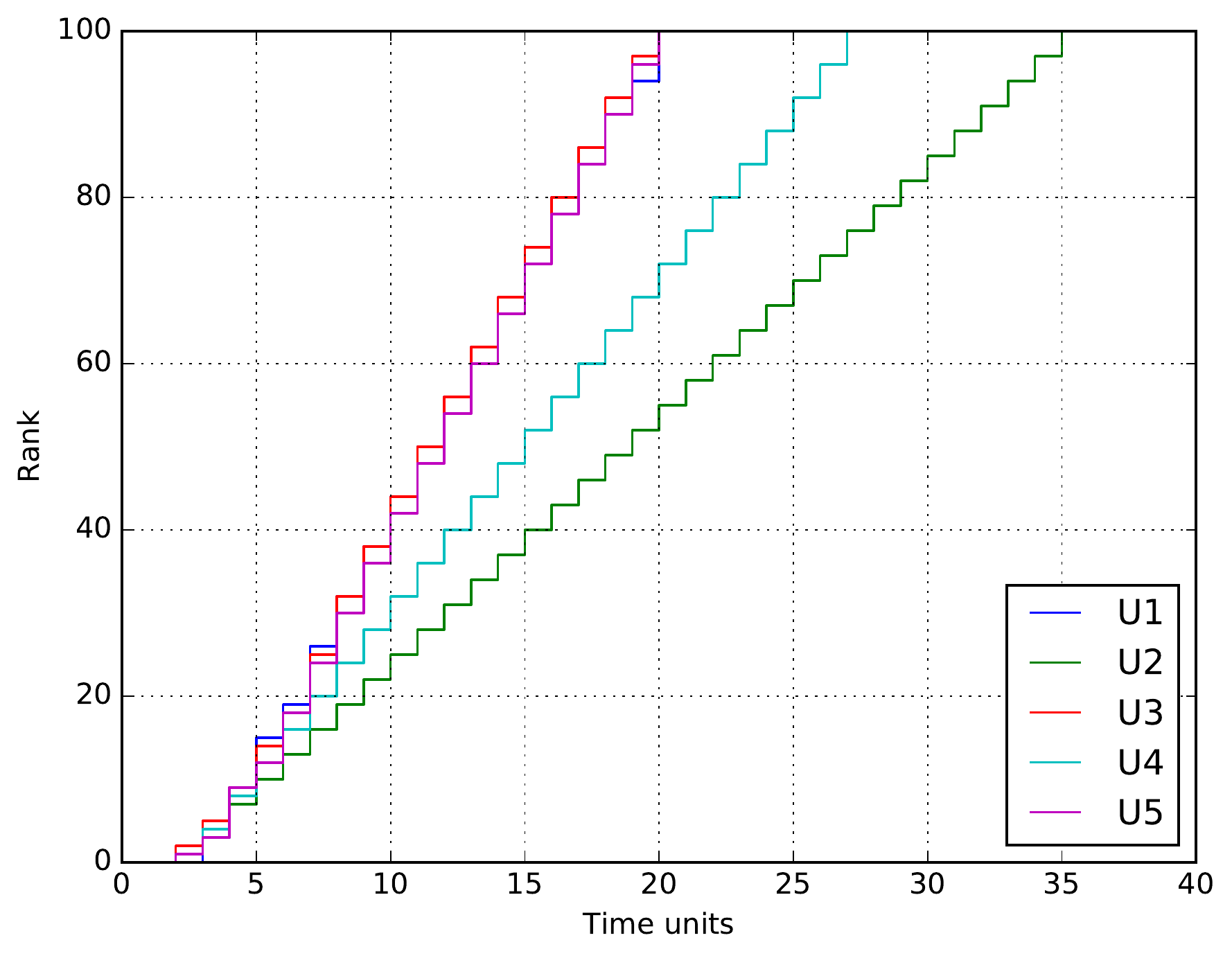}}

\subfloat[MICN]{\centering{}\includegraphics[width=0.50\columnwidth]{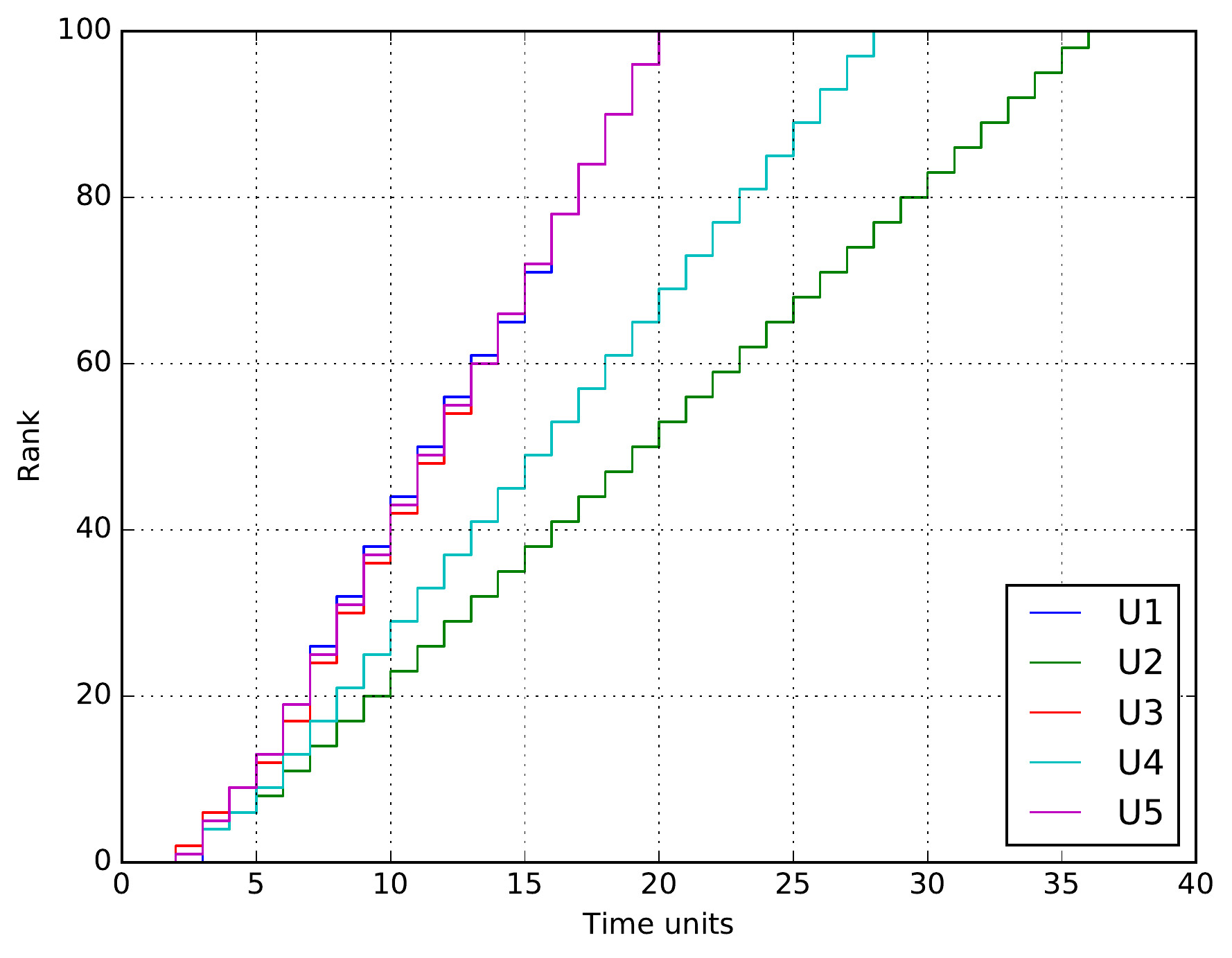}}

\subfloat[MICN-IC]{\centering{}\includegraphics[width=0.50\columnwidth]{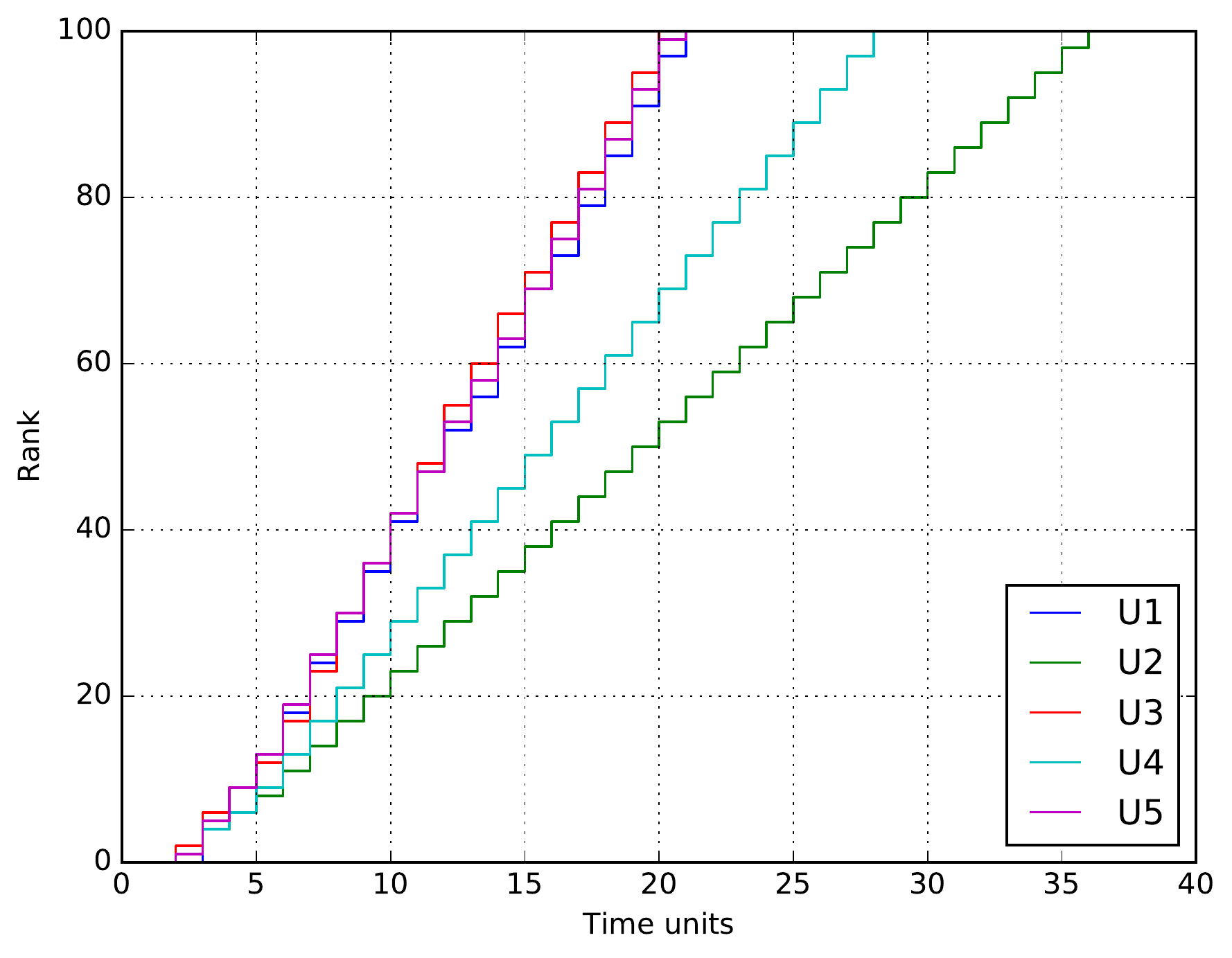}}
\par\end{centering}
\centering{}\caption{PlanetLab topology: Rank evolution as a function of time \label{fig:Rank-plab}}
\end{figure}
\par\end{center}

As seen in Fig.~\ref{fig:Rank-plab}, with MICN, MICN-IC, and NetCodCCN,
clients receive enough content to decode a generation at a rate above
95\% of the maximum rate (provided by the min-cut between the source
and the clients), as observed for the butterfly topology.

Fig.~\ref{fig:traffic-plab} illustrates that the cumulative number
of data packets transmitted on all the links of the network as a function
of time. NetCodCCN generates the most data traffic (also for the longer
duration), followed by MICN. MICN-IC performs the best in terms of
traffic, with respectively 2.16 and 3.42 times fewer transmitted data
packets compared to MICN and NetCodCCN. In the PlanetLab topology,
with the considered scenario, the amount of non-innovative packets
dominates: about 80\% of the content traffic with NetCodCCN is redundant
(non-innovative). Some innovative packets might not be useful for
the client because intermediate nodes of the network are unable to
detect when the client has received all packets required to decode
a generation.
\begin{figure}[H]
\centering
\subfloat[NetCodCCN]{\centering{}\includegraphics[width=0.50\columnwidth]{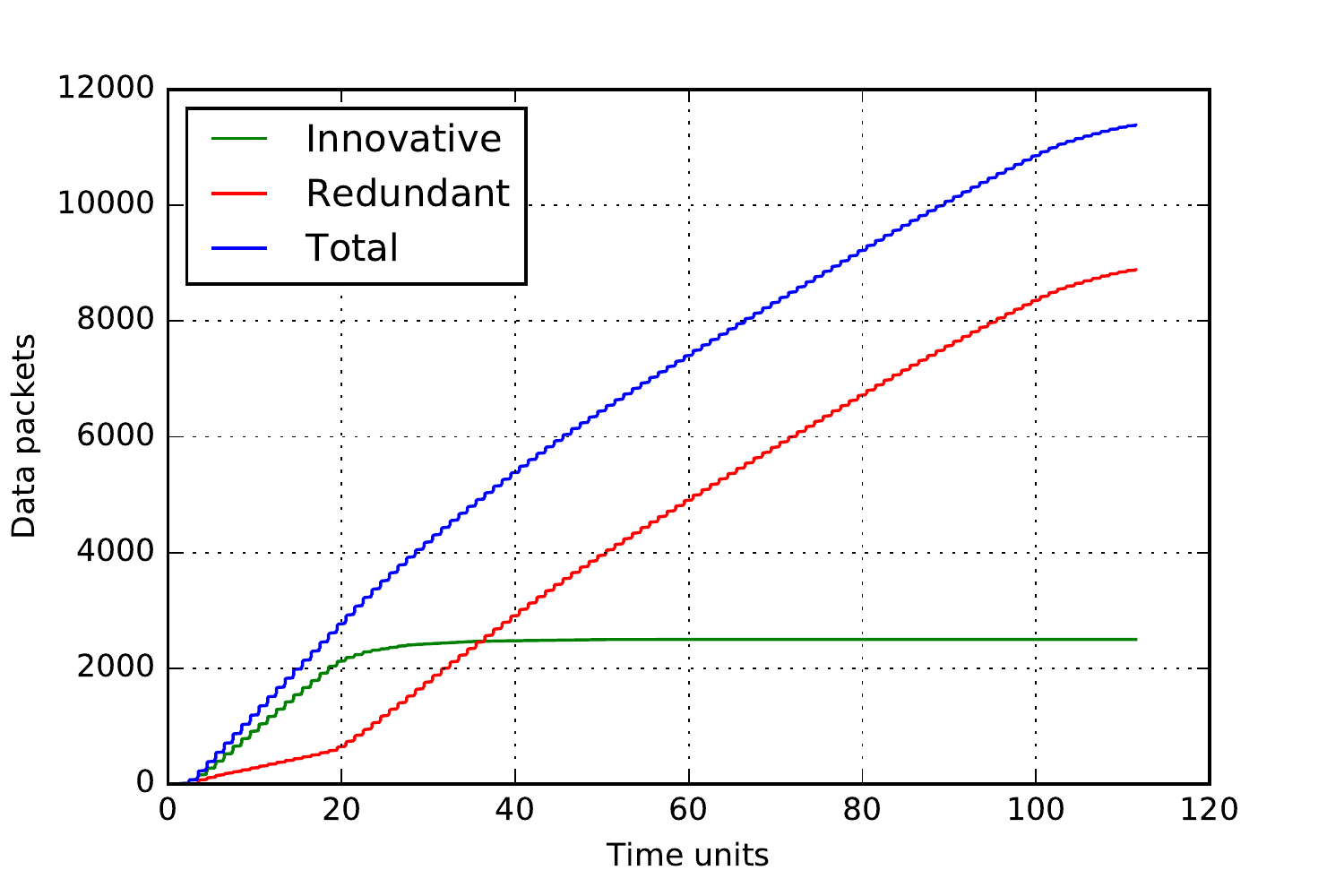}}

\subfloat[MICN]{\centering{}\includegraphics[width=0.50\columnwidth]{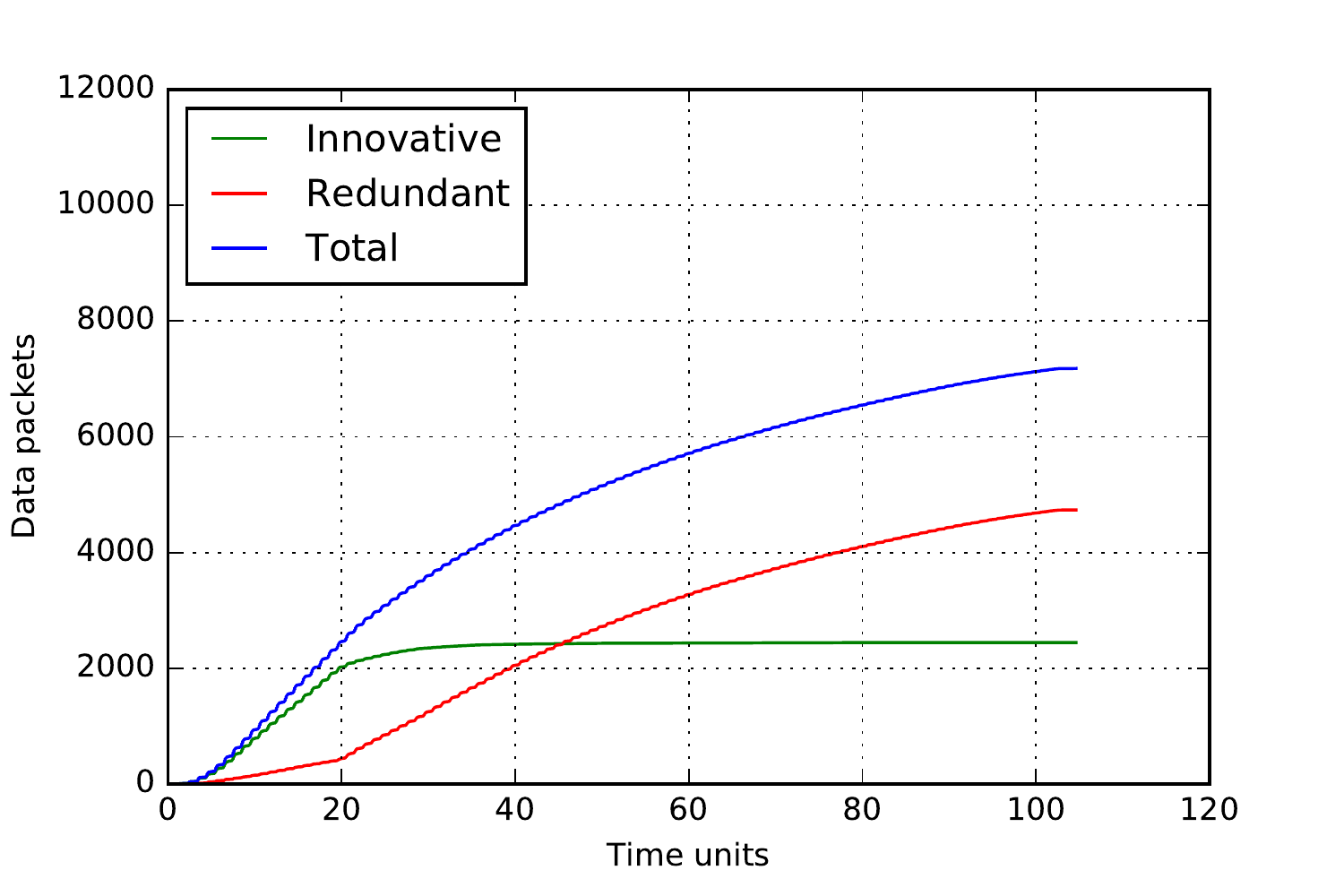}}

\subfloat[MICN-cancel]{\centering{}\includegraphics[width=0.50\columnwidth]{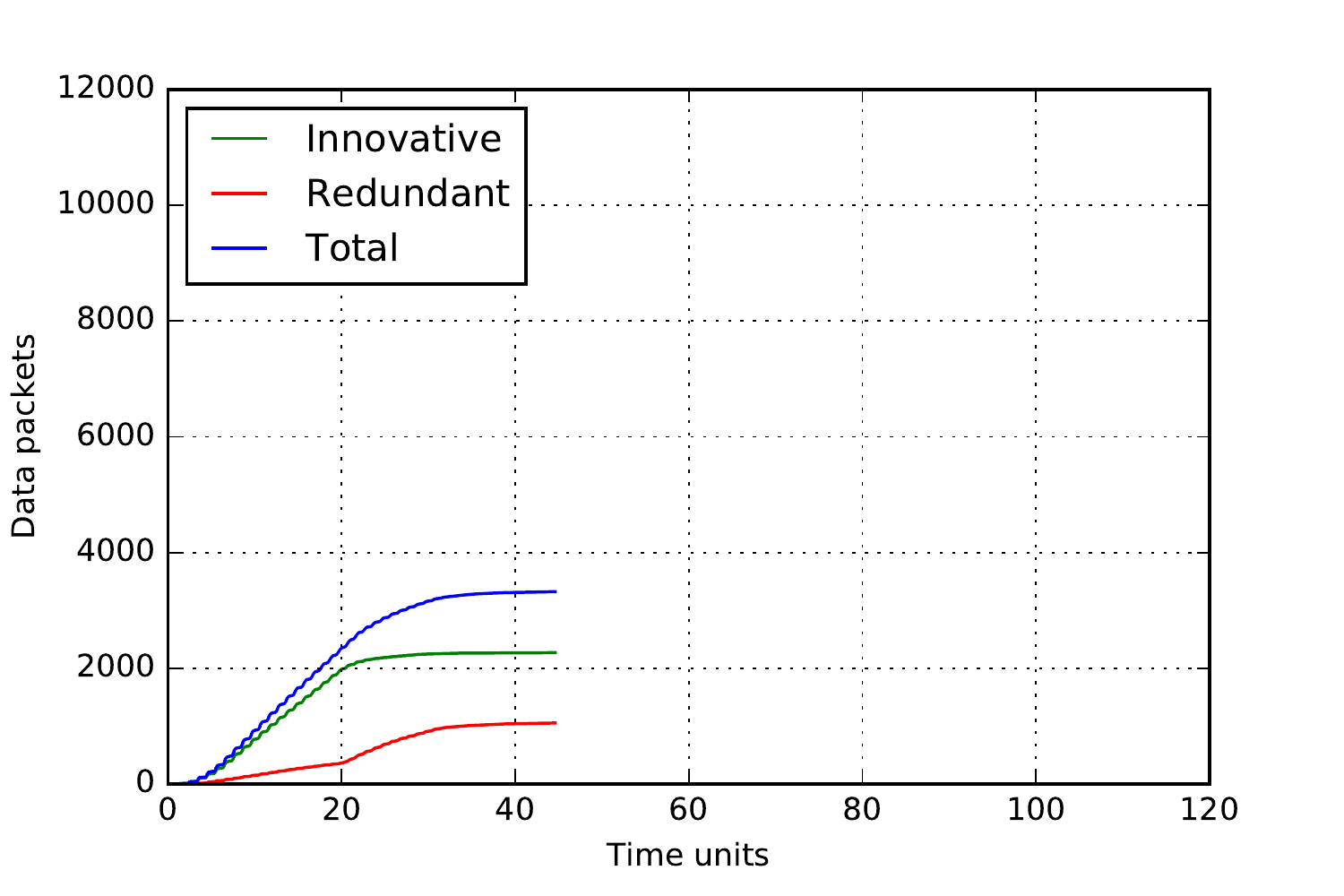}}
\caption{Data traffic in PlanetLab topology \label{fig:traffic-plab}}
\end{figure}
In the PlanetLab topology, the pipeline size impacts the performance
only when it is too small, as shown in Fig.~\ref{fig:pipline-plab}.
Increasing the pipeline size above 2 (MICN), 3 (MICN-IC), and 5 (NetCodCCN)
do not bring additional benefit.
\begin{figure}[H]
\centering
\includegraphics[width=0.5\columnwidth]{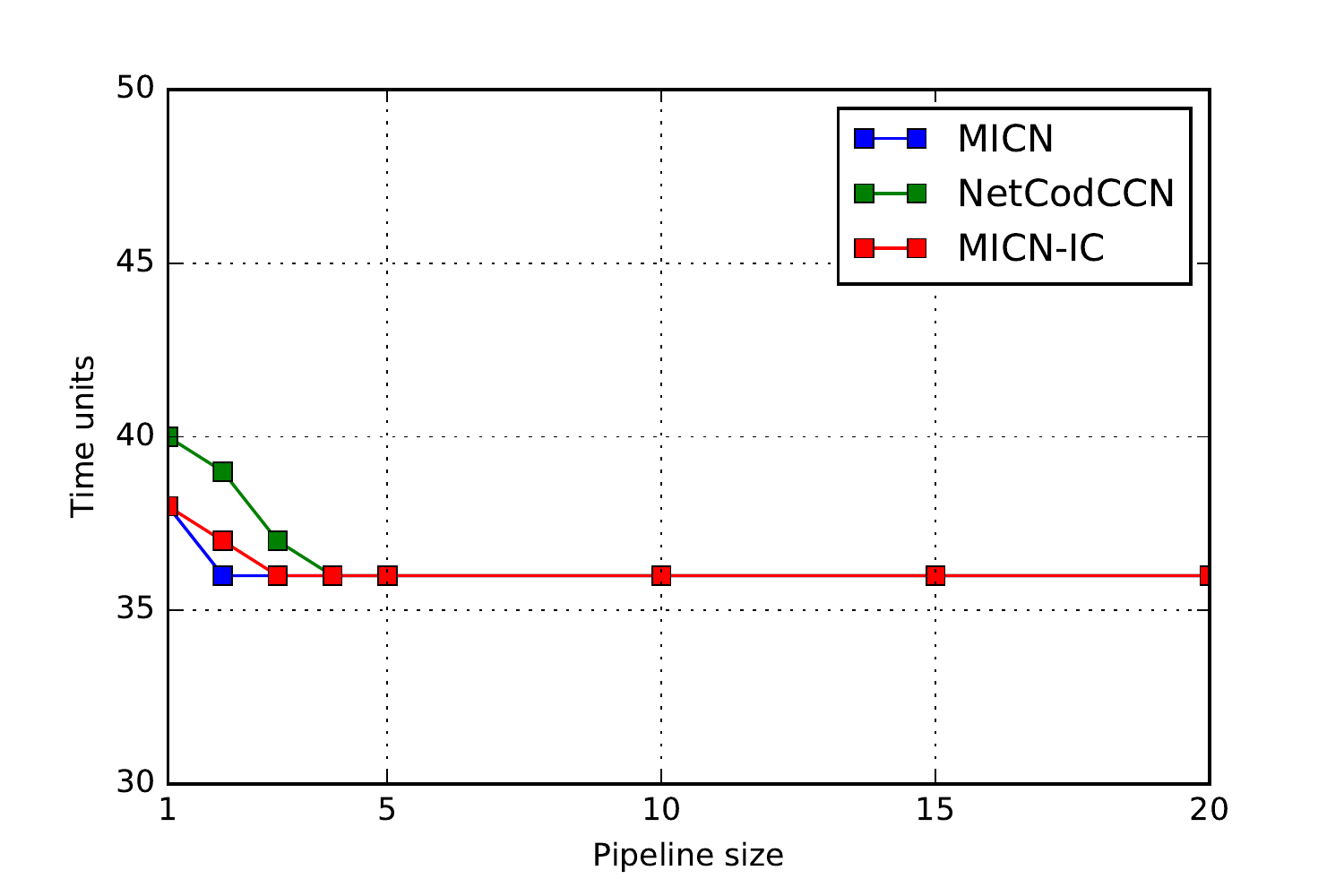}
\caption{\label{fig:pipline-plab}PlanetLab topology: Download time vs pipeline
size, no losses}
\end{figure}

Fig.~\ref{fig:losses_plab} shows the effect of transmission losses.
The download time with MICN and MICN-IC increases almost linearly
with the loss rate, compared to NetCodCCN, which increases faster
when the loss rate is above $10$\%. In the PlanetLab topology, compared
to the butterfly topology, MICN, MICN-IC, and NetCodCCN are all more
robust to packet losses due to the more significant amount of redundant
content traffic in the network which helps to compensate for the losses.

\begin{figure}[H]
\centering{}\includegraphics[width=0.5\columnwidth]{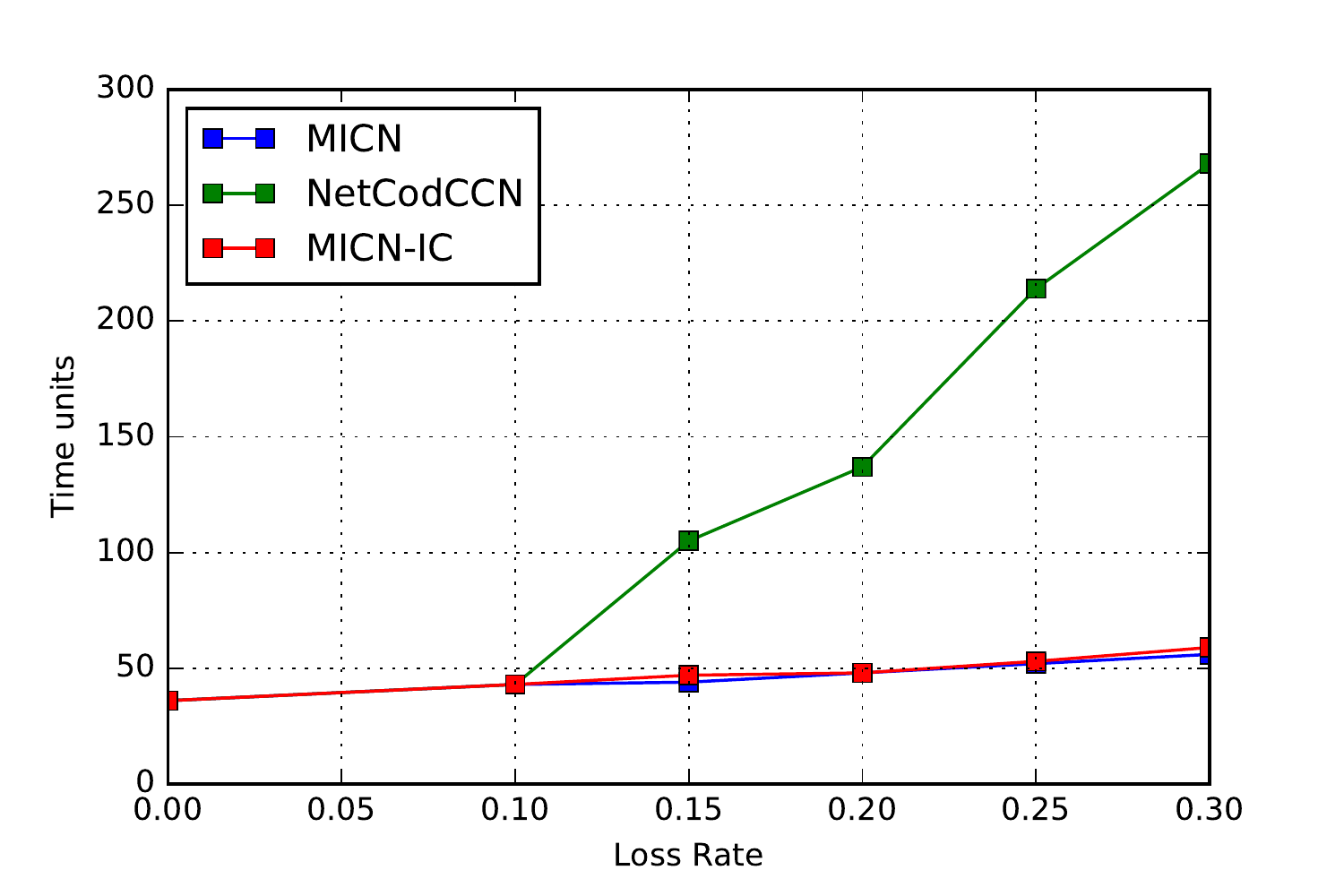}\caption{\label{fig:losses_plab}PlanetLab Topology:Download time vs transmission
error rate}
\end{figure}

\section{Conclusion}

\label{sec:Conclusion}

In this work, we propose a novel way of integrating NC and information-centric
networking. The proposed MICN protocol is built around the MILIC construction
that allows the clients to request content that belongs to predefined
subsets by adding an index in the interest, that indicates the subset.
This interest naming allows the nodes to send multiple interests in
parallel and ensures that linearly independent content satisfies each
interest. In the considered scenarios, the clients download content
close to their maximum capacity (like NetCodCCN). Nevertheless, thanks
to interest cancellation, MILIC-IC limits the redundant data traffic
considerably. This reduces the network load and leaves earlier free
network resources to fetch contents from consecutive generations.

Our future research includes investigating improved interest forwarding
algorithms to use the multiple active links better and reduce the
data traffic by adjusting the number of outgoing interests.

\subsection*{Acknowledgements}

This research was partly supported by Labex DigiCosme (project ANR11 LABEX0045DIGICOSME) operated by ANR as part of the program « Investissement d'Avenir » Idex Paris-Saclay (ANR11IDEX000302).

\addcontentsline{toc}{section}{\refname}\bibliography{library}

\end{document}